\documentclass[10pt,twocolumn]{IEEEtran}
\usepackage{ifpdf}
\usepackage{manfnt}
\usepackage[mathscr]{eucal}
\usepackage{graphicx}
\newcommand{\figsize}{0.92\linewidth}

\usepackage{caption}
\usepackage{algorithmicx}
\usepackage[ruled,vlined,commentsnumbered]{algorithm2e}
\usepackage{amssymb}
\usepackage{theorem}
\usepackage{caption}
\usepackage{subcaption}
\usepackage{cite}
\usepackage{color}
\usepackage{colortbl}
\definecolor{colorref}{rgb}{0.4648,0,0} 
\definecolor{colorcite}{rgb}{0,0.2902,0.1765}

\usepackage{bm}

\usepackage{amsmath}
\DeclareMathOperator*\argmax{arg \, max}		

\newtheorem{proposition}{Proposition}

\newtheorem{lemma}{Lemma}

\newtheorem{Cor}{Corollary}

\newcommand{\setn}{\mathcal{N}}
\newcommand{\setm}{\mathcal{M}}

\newcommand{\Rmnum}[1]{\uppercase\expandafter{\romannumeral #1}}

\renewcommand{\baselinestretch}{1.33}

\newcommand{\psignal}[1]{\boldsymbol{#1}^\prime}

\newcommand{\signal}[1]{{\boldsymbol{#1}}}

\newcommand{\dprime}{{\prime\prime}}

\newcommand{\real}{{\mathbb R}}

\newtheorem{definition}{Definition}

\newtheorem{fact}{Fact}

\newtheorem{problem}{Problem}
\newtheorem{example}{Example}
\newcommand{\Natural}{{\mathbb N}}
\newcommand{\refeq}[1]{(\ref{#1})}

\newif\iftodo   
\todotrue
\newif\iftodoshort  
\todoshortfalse

\newenvironment{proof}{{\it Proof:}}{\hfill$\square$\\}
\renewcommand{\figsize}{0.9\linewidth}

\begin{document}

\title{Fundamental properties of solutions to utility maximization problems in wireless networks}

\author{ R.~L.~G.~Cavalcante and S.~Sta\'nczak \\ {Fraunhofer Heinrich Hertz Institute and Technical University of Berlin \\ 
	{\small email: \{renato.cavalcante,~slawomir.stanczak\}@hhi.fraunhofer.de } }}

\maketitle
\setlength{\belowdisplayskip}{3pt} \setlength{\belowdisplayshortskip}{3pt}
\setlength{\abovedisplayskip}{3pt} \setlength{\abovedisplayshortskip}{3pt}

\begin{abstract} 
	We introduce a unified framework for the study of the utility and the energy efficiency of solutions to a large class of weighted max-min utility maximization problems in interference-coupled wireless networks. In more detail, given a network utility maximization problem parameterized by a maximum power budget $\bar{p}$ available to network elements, we define two functions that map the power budget $\bar{p}$ to the energy efficiency and to the utility achieved by the solution. Among many interesting properties, we prove that these functions are continuous and monotonic. In addition, we derive bounds revealing that the solutions to utility maximization problems are characterized by a low and a high power regime. In the low power regime, the energy efficiency of the solution can decrease slowly as the power budget increases, and the network utility grows linearly at best. In contrast, in the high power regime, the energy efficiency typically scales as $\Theta(1/\bar{p})$ as $\bar{p}\to\infty$, and the network utility scales as $\Theta(1)$. We apply the theoretical findings to a novel weighted rate maximization problem involving the joint optimization of  the uplink power and the base station assignment. 
\end{abstract}

\section{Introduction}
To cope with the ever increasing rate demand of wireless networks in a cost effective way, system engineers need to improve the energy efficiency, which often translates to increasing the rates for a given power budget. This fact has motivated many studies on trade-offs between achievable rates and energy efficiency for many years \cite{Gallager88,Verdu02,renato14SPM}. In particular, the field of information theory has been fundamental to reveal bounds that cannot be exceeded irrespective of the available computational power \cite{Gallager88,Verdu02}. Unfortunately, extending existing information theoretic results to general wireless networks, while capturing limitations of practical hardware and communication strategies, has been proven notoriously difficult. However, as we show in this study, useful and surprisingly simple performance bounds for a large class of communication strategies in wireless networks are available if we depart from the formal setting of information theory.

In practical wireless systems, the parameters of a network configuration are often obtained by solving optimization problems \cite{nuz07,chiang2008power,slawomir09,martin11,renato2016maxmin,cai2011,cai2012max,huang2013joint,sun2014,hong2014,tan2016optimal,foschini1993}. In particular, it is well-known that many weighted max-min rate or signal-to-interference-plus-noise ratio (SINR) maximization problems can be posed as conditional eigenvalue problems involving nonlinear mappings  \cite{nuz07,renato2016maxmin,cai2011,cai2012max,huang2013joint,sun2014,hong2014,tan2016optimal}. The practical implication of this observation is that these maximization problems can be solved with simple iterative fixed point algorithms similar to the standard power method in linear algebra \cite{krause1986perron,krause01}. One of the first studies to establish the connection between nonlinear conditional eigenvalue problems and utility maximization in wireless networks is shown in \cite{nuz07}. Later results appeared  in, to cite a few, \cite{renato2016maxmin,cai2011,cai2012max,huang2013joint,sun2014,hong2014}, which considered utility optimization problems assuming different interference models.

Building upon the findings in \cite{nuz07}, we start by explicitly stating a canonical problem that is solved in many of the applications addressed in  \cite{renato2016maxmin,cai2011,cai2012max,huang2013joint,sun2014,hong2014,nuz07}. Unlike these previous studies, which mostly focus on developing efficient numerical solvers or on posing the utility maximization problems as conditional eigenvalue problems, the objective of this study is to derive properties of the solutions to the canonical problem. Particular emphasis is devoted to properties that provide us with highly valuable insights into the energy efficiency and the utility of networks.

In more detail, given the large class of transmission strategies covered by the canonical problem, we can only evaluate the energy efficiency or the utility achieved by the solution after solving the canonical problem with iterative algorithms. This process can be time consuming, so we exploit  properties of the solution to conditional eigenvalue problems and results on asymptotic or recession functions in convex analysis \cite{aus03} to derive simple and useful bounds on the network utility. These bounds are then used to derive novel bounds on the energy efficiency. 

The above results reveal interesting phenomena (some already observed in particular interference models \cite{song2007network}) that are common to all network utility maximization problems that can be written in the canonical form shown here. More specifically, the solutions are typically characterized by two power regimes: a low power regime and a high power regime. In both regimes, the network utility and the energy efficiency are always monotonically increasing and non-increasing, respectively, as a function of the power budget $\bar{p}$ available to the transmitters. However, in the low power regime, the energy efficiency is bounded by a constant, and it can decrease slowly as we increase the power budget. In contrast, the network utility is upper bounded by a linear function. In the high power regime, the energy efficiency shows a fast decay because it typically scales as $\Theta(1/\bar{p})$ as $\bar{p}\to\infty$, whereas gains in network utility saturate because the network utility scales as $\Theta(1)$ as $\bar{p}\to\infty$ (see Sect.~\ref{sect.preliminaries} for the definition of $\Theta$). The bounds derived here do not depend on any unknown constants, so the power budget characterizing the boundary of the power regions is precisely known. In addition, we show that the spectral radius of lower bounding matrices (a concept introduced in \cite{renato2016}) provides us with a formal means of identifying {\it interference limited networks}, which we define as networks for which the utility cannot grow unboundedly as the power budget diverges to infinity. We also use the concept of recession functions in convex analysis to characterize networks for which the utility can grow unboundedly with increasing power budget, and we call these networks {\it noise limited networks}. 

We illustrate the above theoretical findings in a novel joint uplink power control and base station assignment problem for weighted rate maximization. This application is related to that in \cite{sun2014}, but here we use results shown in \cite{cai2012optimal}, which have been independently obtained in \cite{renato14SIP,renato2016power} in the context of load coupled interference models, to pose the optimization problem in terms of achievable rates instead of the SINR. As a result, we work directly with the variables of interest to system designers (in contrast, note that maximizing the SINR is only an indirect approach to the problem of improving the rates). We emphasize that solving weighted rate maximization problems by choosing appropriate coefficients for weighted SINR maximization problems may not be straightforward because the bijective relation between rate and SINR used in many studies is not affine. One interesting consequence of our novel formulation is that the simple solver based on the fixed point algorithm in \cite{krause1986perron,krause01} becomes readily available. Furthermore, this application exemplifies the validity of the theoretical findings with interference models based on concave mappings that are neither affine nor differentiable. 

This study is structured as follows. In Sect.~\ref{sect.preliminaries} we review definitions and known mathematical tools that are extensively used to prove the main results in this study. In Sect.~\ref{sect.ee} we introduce a new framework for the study of the energy efficiency and the achievable utility of solutions to a large class of network utility maximization problems. The general results obtained in Sect.~\ref{sect.ee} are illustrated with a concrete application in Sect.~\ref{sect.examples}.

\section{Preliminaries}
\label{sect.preliminaries}
 The intention of this section is twofold. First, we try to make this study as self-contained as possible by presenting many standard definitions and results that are essential for the proofs in the next sections. Second, we clarify much of the notation used throughout this study. We note that much of the background material collected here has been taken directly from \cite{renato2016,renato2016maxmin}. In this section, we also show the first (minor) technical result (Proposition~\ref{proposition.affine_bound}).

In more detail, for given $(\signal{x},\signal{y})\in\real^N\times\real^N$, the inequality $\signal{x}\le\signal{y}$ should be understood as a entry-wise inequality. The transpose of vectors or matrices is denoted by $(\cdot)^t$. The sets $\real_+$ and $\real_{++}$ are the sets of non-negative reals and positive reals, respectively. The spectral radius of a matrix $\signal{M}\in\real^{N\times N}$ is denoted by $\rho(\signal{M})$. The effective domain of a function $f:\real^N\to\real\cup\{-\infty\}$ is $\mathrm{dom} f :=\{\signal{x}\in\real^N~|~f(\signal{x})>-\infty\}$. Given two functions $f:\real_+\to\real_+$ and $g:\real_+\to\real_+$ we say that $f$ scales as $\Theta(g(x))$ when $x\to\infty$ (or, in set notation, $f(x)\in\Theta(g(x))$ as $x\to\infty$) if 
\begin{multline*}
(\exists k_0\in\real_{++}) (\exists k_1\in\real_{++}) (\exists k_3\in\real_{++})(\forall x\in\real_+) \\
 x\ge k_0 \Rightarrow k_1 g(x)\le f(x)\le k_2 g(x).
\end{multline*}
If $g$ is a constant function, then we use the convention $f(x)\in\Theta(1)$.

We use the notation ${\mathrm{conv}}~C$ to indicate the \emph{convex hull} of $C\subset\real^N$; i.e., the smallest convex subset of $\real^N$ containing $C$ \cite[p. 43]{baus11}. The \emph{interior} of $C \subset \real^N$ is the set given by  $\mathrm{int}~C:=\{\signal{x}\in C~|~(\exists\epsilon\in\real_{++})~B(\signal{x}; \epsilon)\subset C\}$ \cite[p. 90]{baus11}, where $B(\signal{x};\epsilon):=\{\signal{y}\in\real^N~|~\|\signal{x}-\signal{y}\|_2\le \epsilon\}$ is the closed ball centered at $\signal{x}$ with radius $\epsilon>0$ and $\|\cdot\|_2$ is the standard Euclidean norm. A set $C\subset\real_{+}^N$ is said to be \emph{downward comprehensible} on $\real_+^N$ if $(\forall\signal{x}\in C)(\forall\signal{y}\in \real^N_+)~ \signal{y}\le \signal{x}\Rightarrow \signal{y}\in C$ \cite[p. 30]{martin11}. A convex set $C\subset\real^N$ is \emph{symmetric} if $\signal{x}\in C$ implies $-\signal{x}\in C$, and it is called a \emph{convex body} if it is a compact convex set with nonempty interior \cite[Ch.~1]{ver2011}.

We say that a mapping $T:\real_+^N \to \real^N$ is concave if 
\begin{multline*}
(\forall \signal{x}\in \real_+^N)(\forall \signal{y}\in \real_+^N)(\forall \alpha\in~]0,1[)\\ 
T(\alpha \signal{x}+(1-\alpha)\signal{y})\ge \alpha T(\signal{x})+(1-\alpha)T(\signal{y}).
\end{multline*} 
As shown below, positive concave mappings $T:\real_{+}^N\to\real_{++}^N$ are instances of standard interference functions, which are functions with many applications in wireless networks \cite{yates95}. A simple proof of the following fact can be seen in \cite{renato2016}, among other studies.

\begin{fact}  
	\label{fact.concave_mappings}
	Let $T:\real_+^N\to\real_{++}^N$ be a concave mapping. Then each of the following holds:
	\begin{enumerate}
		\item ({\it Scalability}) $(\forall\signal{x}\in\real^N_+)$ ($\forall \mu >1$) $\mu {T}(\signal{x})>T(\mu\signal{x})$. \par
		\item ({\it Monotonicity}) $(\forall\signal{x}_1\in\real^N_+)$ $(\forall\signal{x}_2\in\real^N_+)$ $\signal{x}_1\ge \signal{x}_2 \Rightarrow T(\signal{x}_1)\ge T(\signal{x}_2)$.
	\end{enumerate}
\end{fact}

In the next section, we extensively exploit the close connection between a large class of utility maximization problems in wireless networks and conditional eigenvalues of concave mappings. Before stating the conditional eigenvalue problem, we introduce the definition of monotone norms used in \cite{krause01,krause1986perron}, and we refer the reader to \cite{john91} for nonequivalent notions of monotonicity that are also common in the literature. 

\begin{definition}
	\label{definition.monotone_norm}
(Monotone norm) A norm $\|\cdot \|$ on $\real^N$ is said to be monotone if 
\begin{align*}
(\forall \signal{x}\in\real^N_+)(\forall \signal{y}\in\real^N_+) \quad \signal{x}\ge\signal{y} \Rightarrow \|\signal{x}\|\ge\|\signal{y}\|.
\end{align*}
\end{definition}

Note that the widely used $l_p$ norms, $1\le p \le \infty$, are monotone in the sense of Definition~\ref{definition.monotone_norm}. 

\begin{fact}
	\label{fact.norm_equiv}
	(\cite[Ch. 13.5]{micheal07} Equivalence of norms in finite dimensional spaces) Let $\|\cdot\|_\alpha$ and $\|\cdot\|_\beta$ be arbitrary norms defined on $\real^N$. Then $(\exists k_1\in\real_{++})(\exists k_2\in\real_{++})(\forall \signal{x}\in\real^N)\quad k_1 \|\signal{x}\|_\alpha \le \|\signal{x}\|_\beta \le k_2 \|\signal{x}\|_\alpha$.
\end{fact}

Recall that a sequence $(\signal{x}_n)_{n\in\Natural}\subset \real^N$ is said to converge to  $\signal{x}\in\real^N$ if the sequence  $(\|\signal{x}_n-\signal{x}\|)_{n\in\Natural}\subset\real_+$ converges to zero (by Fact~\ref{fact.norm_equiv}, the convergence holds for any choice of the norm $\|\cdot\|$). In this case, we write $\lim_{n\to\infty}\signal{x}_n=\signal{x}$. We can now formally introduce the conditional eigenvalue problem and a simple iterative solver:

\begin{fact}
\label{fact.krause} 
(\cite{krause1986perron,krause01}) Let $T:\real_{+}^N\to\real_{++}^N$ be a concave mapping and $\|\cdot\|$ a monotone norm. Then each of the following holds:
\begin{enumerate}
\item There exists a unique solution $(\signal{x}^\star, \lambda^\star)\in\real_{++}^N\times\real_{++}$ to the conditional eigenvalue problem
\begin{problem}
	\label{problem.cond_eig}
	Find $(\signal{x}, \lambda)\in\real_{+}^N\times\real_{+}$ such that $T(\signal{x})=\lambda\signal{x}$ and $\|\signal{x}\|=1$.
\end{problem}
For reference, the scalar $\lambda^\star$ is said to be the conditional eigenvalue of $T$ for the norm $\|\cdot\|$.
\item The sequence $(\signal{x}_n)_{n\in\Natural}\subset\real_{+}^N$ generated by 
\begin{align}
\label{eq.krause_iter}
\signal{x}_{n+1} = T^\prime({\signal{x}_n}):=\dfrac{1}{\|T(\signal{x}_n)\|}T(\signal{x}_n),\quad\signal{x}_1\in\real_{+}^N,
\end{align}
converges geometrically to the uniquely existing vector $\signal{x}^\star\in\mathrm{Fix}(T^\prime):=\{\signal{x}\in\real_{+}^N~|~\signal{x}=T^\prime(\signal{x})\}$, which is also the vector $\signal{x}^\star$ of the tuple $(\signal{x}^\star,\lambda^\star)$ that solves Problem~\ref{problem.cond_eig}. Furthermore, the sequence $(\lambda_n:=\|T(\signal{x}_n)\|)_{n\in\Natural}\subset\real_{++}$ satisfies  $\lim_{n\to\infty}\lambda_n=\lambda^\star$.
\end{enumerate}
\end{fact}

To find a simple lower bound for the conditional eigenvalue $\lambda^\star$ of a concave mapping $T$ for any monotone norm, we can use the concept of lower bounding matrices introduced in \cite{renato2016}. 

\begin{definition}
\label{def.bmatrix}
(\cite{renato2016} Lower bounding matrices) Let $T:\real_{+}^N\to\real_{++}^N$ be a continuous concave mapping, and denote by $t_1:\real_+^N\to\real_{++},~\ldots, t_N:\real_+^N\to\real_{++}$ the continuous concave functions such that $[t_1(\signal{x}),\ldots,t_N(\signal{x})]:={T}(\signal{x})$ for every $\signal{x}\in\real_{+}^N$. The lower bounding matrix $\signal{M}\in\real_{+}^{N\times N}$ of the mapping $T$ is the matrix with its $i$th row and $j$th column given by\footnote{See \cite[Example 1]{renato2016} for an alternative and equivalent construction method based on supergradients.}
\begin{align}
\label{eq.lbmatrix}
[\signal{M}]_{i,j} = \lim_{h\to 0^+} h~t_i(h^{-1}\signal{e}_j) \in \real_+,
\end{align}
where the limit can be shown to exist and $\signal{e}_j$ ($j\in\{1,\ldots,N\}$) is the unit vector with all components being zero, except for the $j$th component, which is equal to one.
\end{definition}

Lower bounding matrices are at the heart of many of the results in this study because of the next result:

\begin{fact}
	\label{fact.eigbound} (\cite{renato2016maxmin})
	Let $\signal{M}\in\real^{N\times N}_{+}$ be the lower bounding matrix of a continuous  concave mapping $T:\real_{+}^N\to\real_{++}^N$. In addition, let $(\signal{x}^\star, \lambda^\star)\in\real_{++}^N\times\real_{++}$ be the solution to Problem~\ref{problem.cond_eig} for an arbitrary monotone norm. Then we have $\lambda^\star > \rho(\signal{M})$.
\end{fact}

By considering Definition~\ref{definition.recession} below, we can observe the strong connection between \refeq{eq.lbmatrix} and the concept of recession or asymptotic functions in convex analysis, which we use to study the behavior of networks in the high power regime. 

\begin{definition}
	\label{definition.recession}
	(Recession or asymptotic functions \cite[Ch. 2.5]{aus03}) Let $f:\real^N\to\real\cup\{-\infty\}$ be upper semicontinuous, proper, and concave. We define its recession or asymptotic function at $\signal{y}\in\real^N$ the function $f_\infty$ given by
	\begin{align*}
	(\forall \signal{x}\in\mathrm{dom}~ f)~~ f_\infty(\signal{y}):=\lim_{h\to\infty}\dfrac{f(\signal{x}+h\signal{y})-f(\signal{x})}{h}.
	\end{align*}
	The above limit can be more conveniently calculated by using \cite[Corollary 2.5.3]{aus03}
	\begin{align}
	\label{eq.recession_func}	
	(\forall \signal{y}\in\mathrm{dom}~ f)\quad 
	f_\infty(\signal{y})=\lim_{h\to 0^+}h f(h^{-1}\signal{y}),
	\end{align}
	and the equality in \refeq{eq.recession_func} is valid for every $\signal{y}\in\real^N$ if $\signal{0}\in\mathrm{dom}~f$. We also recall that asymptotic functions are \emph{positively homogeneous} \cite[Proposition 2.5.1(a)]{aus03}.
\end{definition}

We end this section with a simple result that is later used for the analysis of the energy efficiency. 
\begin{proposition}
	\label{proposition.affine_bound}
	Let $\|\cdot\|$ be an arbitrary norm and $T:\real_{+}^N\to\real_{++}^N$ an arbitrary continuous concave mapping. Then the following holds:
	\begin{align*}
	(\forall\signal{x}\in\real_{++}^N)(\exists k_1\in\real_{++})(\exists k_2\in\real_{++})(\forall\bar{p}\in\real_{+}) \\ 
	\|T(\bar{p}\signal{x})\| \le k_1\bar{p}+k_2.
	\end{align*}
\end{proposition}
\begin{proof}
Fix $\signal{x}\in\real_{++}^N$ arbitrarily. Denote by $t_1:\real_+^N\to\real_{++},~\ldots, t_N:\real_+^N\to\real_{++}$ the continuous concave functions such that $[t_1(\signal{p}),\ldots,t_N(\signal{p})]:={T}(\signal{p})$ for every $\signal{p}\in\real_{+}^N$, and note that, for each $i\in\mathcal{N}:=\{1,\ldots,N\}$, the function $g_i:\real_{+}\to\real_{++}:\bar{p}\mapsto t_i(\bar{p}\signal{x})$ is also concave. Now choose $c\in\real_{++}$ arbitrarily. By \cite[Corollary 16.15]{baus11}, we know that $(\forall i \in \setn)~ \partial g_i(c) \neq \emptyset$, where $\partial g_i(c):=\{u\in\real~|~(\forall y\in\real_{+}) (y-c)u + f(c)\ge g_i(y) \}$ is the superdifferential of the concave function $g_i$ at $c$. For each $i\in\setn$, choose an arbitrary tuple $(\alpha_i,\alpha^\prime_i)\in\real\times\real$ satisfying  $\alpha_i>\alpha_i^\prime\in \partial g_i(c)$, and define $\beta_i:=g_{i}(c)>0$. The definition of the superdifferential and the fact that $\alpha_i^\prime\in\real_{+}$ (because $g_i$ is positive and concave; see, e.g.,  \cite[Lemma~1.1]{renato2016}) for every $i\in\setn$ shows that 
\begin{align}
\label{eq.ti_gi}
(\forall i\in\setn) (\forall \bar{p}\in\real_{+})~~ 0<t_i(\bar{p}\signal{x})=g_i(\bar{p}) \le \alpha_i \bar{p}+\beta_i.
\end{align}
 Let $(k_1^\prime, k_2^\prime)\in\real_{++}\times\real_{++}$ satisfy $k_1^\prime \ge \|\signal{\alpha}\|_m$ and $k_2^\prime \ge \|\signal{\beta}\|_m$, where $\signal{\alpha}:=[\alpha_1,\ldots,\alpha_N]\in\real_{++}^N$, $\signal{\beta}:=[\beta_1,\ldots,\beta_N]\in\real_{++}^N$, and $\|\cdot\|_m$ is an arbitrary monotone norm. By \refeq{eq.ti_gi}, the monotonicity of the norm $\|\cdot\|_m$, and the triangle inequality, we deduce
\begin{multline*}
(\forall\bar{p}\in\real_{+}) \|T(\bar{p}\signal{x})\|_m \le \|\bar{p}\signal{\alpha}+\signal{\beta}\|_m \\ \le \bar{p}\|\signal{\alpha}\|_m+\|\signal{\beta}\|_m\le k_1^\prime \bar{p} + k_2^\prime.
\end{multline*}
Since norms are equivalent in finite dimensional spaces (Fact~\ref{fact.norm_equiv}), we obtain
\begin{multline*}
(\exists\alpha\in\real_{++})(\forall\bar{p}\in\real_{+}) \\ \alpha  \|T(\bar{p}\signal{x})\| \le \|T(\bar{p}\signal{x})\|_m \le k_1^\prime \bar{p} + k_2^\prime,
\end{multline*}
which completes the proof as $\signal{x}\in\real_{++}^N$ is arbitrary.
\end{proof}
\section{Proposed framework}
\label{sect.ee}

\subsection{The canonical utility maximization problem}

In this study, we are interested in utility maximization problems that can be posed in the following canonical form:

\begin{problem}
\label{problem.canonical}
(Canonical form of the network utility maximization problem)
\begin{align}
\label{eq.canonical}
\begin{array}{lll}
\mathrm{maximize}_{\signal{p}, c} & c \\
 \mathrm{subject~to} & \signal{p}\in \mathrm{Fix}(cT):=\left\{\signal{p}\in\real_{+}^N~|~\signal{p}=cT(\signal{p})\right\} \\
 & \|\signal{p}\|_a \le \bar{p} \\
 & \signal{p}\in\real_{+}^N, c\in\real_{++},
\end{array}
\end{align}
where $\bar{p}\in \real_{++}$ is a design parameter hereafter called power budget,  $\|\cdot\|_a$ is an arbitrary monotone norm, and $T:\real_{+}^N\to\real_{++}^N$ is an arbitrary {\bf continuous} concave mapping called interference mapping in this study. 
\end{problem}

Following standard terminology, we say that a tuple $(\signal{p}, c)\in\real^N_{++}\times\real_{++}$ is feasible to Problem~\ref{problem.canonical} if it satisfies all constraints in \refeq{eq.canonical}. The set of all feasible tuples is defined to be the feasible set. If a feasible tuple is also a solution to Problem~\ref{problem.canonical}, then we say that this tuple is optimal.

 Problem~2 can be seen as a particular instance of that addressed in \cite{nuz07} (see \cite{cai2012optimal,tan2014wireless,zheng2016} for other notable extensions). However, Problem~\ref{problem.canonical} already covers a large array of network utility optimization problems, and, as we show below, its solution has a rich structure that, to the best of our knowledge, we explore for the first time here. Particular instances of Problem~\ref{problem.canonical} include max-min rate optimization in load coupled networks \cite{renato2016maxmin}, the joint optimization of the uplink power and the cell assignment \cite{sun2014}, the optimization of the uplink receive beamforming \cite[Sect.~1.4.2]{martin11}, and many of the applications described in \cite{cai2012optimal,tan2014wireless,zheng2016}. Later in Sect.~\ref{sect.examples} we show a novel weighted rate maximization problem that is also an instance of Problem~\ref{problem.canonical}.

  Typically, in network utility maximization problems written in the canonical form shown above, the optimization variable $\signal{p}$ corresponds to the transmit power of network elements (e.g., base stations or user equipment); the optimization variable $c$, hereafter called {\it utility}, is the common desired rate, or, depending on the problem formulation, the common SINR of all users; the norm $\|\cdot\|_a$ is chosen based on the energy source of the network elements (e.g., we can use the $l_1$ norm if all networks elements share the same source, or the $l_\infty$ norm if the network elements have independent sources); and $T$ is a known  mapping that captures the interference coupling among network elements. In particular, the mapping $T$ is constructed with information about many environmental and control parameters such as the pathgains, MIMO beamforming techniques, the user-base station assignment mechanisms, and the system bandwidth, to name a few. \emph{For concreteness, in the next sections  we use the above interpretation of the optimization variables to explain in words the implications of the main results in this study}. 
  
  In the next proposition, we show that the seemingly simple power constraints in Problem~\ref{problem.canonical} are equivalent to a rich class of constraints commonly found in applications in wireless networks. The next result is also useful to identify utility optimization problems that cannot be addressed with the formulation in Problem~\ref{problem.canonical}, in which case approaches such as those described in \cite{zheng2016} should be considered.

  \begin{proposition}
  	\label{proposition.mnorm_rep}
  	Let $C\subset\real_+^N$ be a compact convex set with nonempty interior. Assume that $C$ is also downward comprehensible, and define $S := {\mathrm{conv}} (C\cup -C)$, where $-C:=\{\signal{x} \in\real^N | -\signal{x}\in C\}$. Then the gauge function or Minkowski functional of $S$ \cite[Ch.~2]{ver2011}\cite[Corollary~2.5.6]{aus03}, denoted by
  	\begin{align}
  	\label{eq.mink}
  	 \|\signal{x}\|_S:=\inf\{\gamma > 0~|~(1/\gamma) \signal{x}\in S \} 
  	 \end{align}
  	for all $\signal{x}\in\real^N$, is a monotone norm. Furthermore, we have  	
  	\begin{align}
  	\label{eq.c_monotone}
  	C:=\{\signal{x} \in \real_+^N~|~\|\signal{x}\|_S \le 1 \}.
  	\end{align}
  	
  	Conversely, given an arbitrary monotone norm $\|\cdot\|$, the set $C:=\{\signal{x} \in\real^N_+~|~ \|\signal{x}\|\le 1\}$ is a compact convex set that is downward comprehensible on $\real_+^N$. Furthermore, $C$ has nonempty interior.
  	
\end{proposition}
  
  \begin{proof}
  	First note that the set $S$ is compact with nonempty interior because it is the convex hull of a compact set with nonempty interior in a finite dimensional space. It is also symmetric by construction, so $S$ is a symmetric convex body. Therefore, by \cite[Proposition~2.1]{ver2011}, the function in \refeq{eq.mink} defines a norm on $\real^N$, and $S$ can be equivalently expressed as the unit ball $\{\signal{x}\in\real^N~|~\|\signal{x}\|_S\le 1 \}=S$ (see also \cite[Corollary~2.5.6]{aus03}). We now proceed to show that this norm is monotone and that $C=\{\signal{x} \in \real_+^N~|~\|\signal{x}\|_S \le 1 \}=\{\signal{x} \in \real_+^N~|~\signal{x}\in S \}$; i.e., given $C$, the operation ${\mathrm{conv}} (C\cup -C)$ does not remove or add any vectors to the nonnegative orthant. 
  	
  	It is clear that $\signal{x}\in C$ implies $\signal{x}\in S$. We now prove that $\signal{x}\in S\cap {\real_+^N}$ implies $\signal{x}\in C$. As a consequence of \cite[Proposition~3.4]{baus11}, any vector $\signal{x}\in S\cap \signal{\real_+^N}$ can be written as 
  	\begin{align*}
  	  	\signal{x} = \sum_{i=1}^{I_1} \alpha^\prime_i \psignal{x}_i + \sum_{i=1}^{I_2} \alpha^\dprime_i \signal{x}^\dprime_i,
  	\end{align*}
  	where $I_1\in\Natural$, $I_2\in\Natural$, $\{\signal{x}^\prime_i\}_{i\in\{1,\ldots,I_1\}} \subset -C$, $\{\signal{x}^\dprime_i\}_{i\in\{1,\ldots,I_2\}} \subset C$, $\{\alpha^\prime_i\}_{i\in\{1,\ldots, I_1\}}\subset\real_+$,  $\{\alpha^\dprime_i\}_{i\in\{1,\ldots, I_2\}}\subset\real_{++}$, and $\sum_{i=1}^{I_1} \alpha^\prime_i + \sum_{i=1}^{I_2} \alpha^\dprime_i = 1$. Since $-\sum_{i=1}^{I_1} \alpha^\prime_i \psignal{x}_i$ is a non-negative vector and  $0<u:=\sum_{i=1}^{I_2} \alpha^\dprime_i\le 1$, we deduce
  	\begin{align*}
  	\signal{0}\le\signal{x} \le  \sum_{i=1}^{I_2} \alpha^\dprime_i \signal{x}^\dprime_i \le \sum_{i=1}^{I_2} \dfrac{\alpha^\dprime_i}{u} \signal{x}^\dprime_i=:\tilde{\signal{x}}.
  	\end{align*}  	
  	Therefore, $\tilde{\signal{x}}$ is a convex combination of points in $C$ and $\signal{x}\le \tilde{\signal{x}}$. Since $C$ is a convex and downward comprehensible set, we have both $\signal{x}\in C$ and $\tilde{\signal{x}}\in C$. Combining everything, we verify that $\signal{x}\in C$ if and only if $\signal{x}\in S\cap\real_+^N$. Since we have already proved that $S$ is the unity ball  with the norm $\|\cdot\|_S$, the result $\signal{x}\in C \Leftrightarrow \signal{x}\in S\cap\real_+^N$ also implies \refeq{eq.c_monotone}. Moreover, by \refeq{eq.mink}, we can write 
  	  	\begin{align}
  	  	\label{eq.eqv_norm}
  	  	(\forall \signal{x}\in\real_+) \|\signal{x}\|_S=\inf\{\gamma > 0~|~(1/\gamma) \signal{x}\in C \}.
  	  	\end{align}
  	The monotonicity of the norm in \refeq{eq.mink} follows from \refeq{eq.eqv_norm} and the fact that, if $(\signal{x},\signal{y})\in\real_+^N\times \real_+^N$ satisfies  $\signal{0}\le({1}/{\gamma}) \signal{x}\le ({1}/{\gamma}) \signal{y} \in C$ for some $\gamma\in\real_{++}$, then we have $({1}/{\gamma}) \signal{x}\in C$ by downward comprehensibility of $C$. 
  	
  	We now proceed to prove the converse. Let $\|\cdot\|$ be a monotone norm. By using standard arguments in convex analysis, we know that $C:=\{\signal{x} \in\real^N_+~|~ \|\signal{x}\|\le 1\}$ is a compact convex set with nonempty interior. We omit the details for brevity. Furthermore, by monotonicity of the norm, we have 
  	$(\forall\signal{x}\in\real^N_+)(\forall\signal{y}\in C) \signal{x}\le\signal{y} \Rightarrow \|\signal{x}\|\le \|\signal{y}\|\le 1$, which implies $\signal{x}\in C$ and concludes the proof that $C$ is downward comprehensible on $\real^N_+$.

  \end{proof}

   A practical implication of  Proposition~\ref{proposition.mnorm_rep} is that, if we are given power constraints of the form $f_1(\signal{p})\le 0,\ldots, f_K(\signal{p})\le 0$ for possibly nonlinear and nondifferentiable convex functions $f_1:\real_+^N\to\real$, ..., $f_K:\real_+^N\to\real$, and if $C= \cap_{i=1}^K\{\signal{p}\in\real_+^N~|~ f_i(\signal{p})\le 0\}$ satisfies the assumptions of the proposition, then we can equivalently represent these constraints by a monotone norm that can be computed as in \refeq{eq.mink}. If the norm in \refeq{eq.mink} is not easy to obtain in closed form, but we can easily verify whether a given point $\signal{x}\in\real_+^N$ satisfies $\signal{x}\in C$, then the norm can be evaluated numerically by using the simple techniques described in \cite[Algorithm~2,Algorithm~3]{zheng2016} (e.g., the bisection algorithm). Therefore, we can construct the sequence described in \refeq{eq.krause_iter}, and, as we show below, the simple fixed point algorithm in \refeq{eq.krause_iter} with the monotone norm $\|\signal{x}\|:=(1/\bar{p}) \|\signal{x}\|_a $, where $\bar{p}>0$, solves Problem~\ref{problem.canonical}:~\footnote{As already mentioned in the Introduction, the focus of this study is on obtaining a deep understanding of properties of the solution to Problem~\ref{problem.canonical}, and not on the development of algorithmic solutions, which has been the focus of previous work.}
   
\begin{fact} (\cite{nuz07} Properties of the solution to Problem~\ref{problem.canonical})
	
	\begin{enumerate}
		\label{fact.equivalence}
		\item Problem~\ref{problem.canonical} has a unique solution $(\signal{p}^\star, c^\star)\in\real^N_{++}\times\real_{++}$ that is also the solution to the conditional eigenvalue problem:
		\begin{problem}
			Find $(\signal{p}^\star, c^\star)\in\real^N_{++}\times\real_{++}$ such that $T(\signal{p}^\star)=\dfrac{1}{c^\star} \signal{p}^\star$ and $\|\signal{p}^\star\|:=(1/\bar{p}) \|\signal{p}^\star\|_a = 1$.
		\end{problem}

	\item Denote by  $(\signal{p}^\star_{\bar{p}},~c^\star_{\bar{p}})\in\real_{++}^N\times\real_{++}$ the solution to Problem~\ref{problem.canonical} for a given power budget $\bar{p}\in\real_{++}$. Then the function $U:\real_{++}\to\real_{++}:\bar{p}\mapsto c^\star_{\bar{p}}$ is  increasing and $P:\real_{++}\to\real_{++}^N:\bar{p}\mapsto \signal{p}^\star_{\bar{p}}$ is increasing coordinate-wise; i.e.,
	\begin{multline*}
	 (\forall \bar{p}^\prime\in\real_{++}) (\forall \bar{p}^\dprime\in\real_{++}) \bar{p}^\prime > \bar{p}^\dprime \\ \Rightarrow U(\bar{p}^\prime) > U(\bar{p}^\dprime)\text{ and } P(\bar{p}^\prime) > P(\bar{p}^\dprime)
	\end{multline*}
		\end{enumerate}
\end{fact}

For reference, we call the functions $U:\real_{++}\to\real_{++}$ and $P:\real_{++}\to\real_{++}^N$ in Fact~\ref{fact.equivalence}.2 the \emph{utility} and \emph{power} functions, respectively. Uniqueness of the solution to Problem~\ref{problem.canonical} also enables us to define a notion of energy efficiency (utility over power) as follows:

\begin{definition}
	\label{def.mm_ee}
	($\|\cdot\|_b$-energy efficiency function) Let $U:\real_{++}\to\real_{++}$ and $P:\real_{++}\to\real_{++}^N$ be, respectively, the utility and power functions. By assuming that $\|\cdot\|_b$ is a monotone norm, the $\|\cdot\|_b$-energy efficiency function is given by $E:\real_{++}\to\real_{++}:\bar{p}\mapsto U(\bar{p})/\|P(\bar{p})\|_b$, or, equivalently, $E:\real_{++}\to\real_{++}:\bar{p}\mapsto 1/\|T(P(\bar{p})) \|_b$, which is an immediate consequence of Fact~\ref{fact.equivalence}.1.
\end{definition}

\begin{example} If the monotone norm $\|\cdot\|_b$ is selected as the $l_1$ norm for a utility maximization problem in which $\signal{p}$ is a vector of transmit power of base stations (in Watts) and the utility $c$ is the common rate achieved by users (in bits/second), then $E(\bar{p})$ shows how many bits each user receives for each Joule spent by a network optimized for the power budget $\bar{p}$. See \cite{renato2016maxmin} for an example of a network utility maximization problem with variables having this interpretation.
\end{example}

With the above definitions, we now proceed to study properties of the functions $U$, $P$, and $E$. The results in the following sections establish the main contribution of this study. 

\subsection{Monotonicity and continuity of the power, utility, and energy efficiency functions}

Fact~\ref{fact.equivalence} shows that the utility and power functions are (coordinate-wise) increasing. However, as we show below, the transmit power always grows faster than the utility. 

\begin{lemma}
	\label{lemma.ee_prop}
	The $\|\cdot\|_b$-energy efficiency function $E:\real_{++}\to\real_{++}$ in Definition~\ref{def.mm_ee} is non-increasing; i.e.,
	\begin{align*}
		(\forall \bar{p}^\prime\in\real_{++})(\forall \bar{p}^\dprime\in\real_{++})~ \bar{p}^\prime > \bar{p}^\dprime \Rightarrow E(\bar{p}^\prime) \le E(\bar{p}^\dprime).
	\end{align*}
\end{lemma}
\begin{proof}
Denote by $(\signal{p}^\prime, c^\prime)\in\real^N_{++}\times\real_{++}$ and $(\signal{p}^\dprime, c^\dprime)\in\real^N_{++}\times\real_{++}$ the optimal tuples to Problem~\ref{problem.canonical} with the power budget $\bar{p}$ set to, respectively, $\bar{p}^\prime\in\real_{++}$ and $\bar{p}^\dprime\in\real_{++}$. From $\bar{p}^\dprime > \bar{p}^\prime$ and Fact~\ref{fact.equivalence}.2, we obtain $\signal{p}^\dprime > \signal{p}^\prime$, which implies $T(\signal{p}^\dprime) \ge T(\signal{p}^\prime)\ge T(\signal{0}) > \signal{0}$ by Fact~\ref{fact.concave_mappings}.2 and positivity of  $T$. Monotonicity of the norm $\|\cdot\|_b$ now shows that 
\begin{multline*}
1/E(\bar{p}^\dprime) = \|T(\signal{p}^\dprime)\|_b \ge \|T(\signal{p}^\prime)\|_b \\ = 1/E(\bar{p}^\prime) \ge \|T(\signal{0})\|_b > 0,
\end{multline*}
which implies the desired result.
\end{proof}

The following proposition establishes the continuity of the power, utility, and energy efficiency functions.

 \begin{proposition}
 	\label{proposition.continuous}
 	Let $(\bar{p}_n)_{n\in\Natural}\subset\real_{++}$ be an arbitrary sequence of power budgets converging to an arbitrary scalar $\bar{p}^\star\in\real_{++}$. Then the power function $P:\real_{++}\to\real_{++}^N$, the utility function $U:\real_{++}\to\real_{++}$, and the $\|\cdot\|_b$-energy efficiency function $E:\real_{++}\to\real_{++}$ satisfy the following:
 	\begin{enumerate}
 		\item[(i)] $\lim_{n\to\infty}P(\bar{p}_n)=P(\bar{p}^\star)$;
 		\item[(ii)] $\lim_{n\to\infty}U(\bar{p}_n)=U(\bar{p}^\star)$; and
 		\item[(iii)] $\lim_{n\to\infty}E(\bar{p}_n)=E(\bar{p}^\star)$.
 	\end{enumerate}
 \end{proposition}
 \begin{proof}
 	To simplify the notation in the proof, we define $(\signal{p}_n,c_n):=(P(\bar{p}_n), U(\bar{p}_n))\in \real_{++}^N\times\real_{++}$ for every $n\in\Natural$ and $(\signal{p}^\star,c^\star):=(P(\bar{p}^\star), U(\bar{p}^\star))\in \real_{++}^N\times\real_{++}$.
 	
 	(i) The sequence $(\bar{p}_n)_{n\in\Natural}$ is bounded because it converges, so  $(\signal{p}_n)_{n\in\Natural}$ is also bounded because $\|\signal{p}_n\|_a = \bar{p}_n$ for every $n\in\Natural$ (see Fact~\ref{fact.equivalence}).  Therefore, by \cite[Lemma~2.38]{baus11} and \cite[Lemma~2.41(ii)]{baus11} (see also \cite[Sect. 2.1]{aus03}), $\lim_{n\to\infty}\signal{p}_n=\signal{p}^\star$ follows if we show that $\signal{p}^\star$ is the only accumulation point of the bounded sequence $(\signal{p}_n)_{n\in\Natural}$.

 	By the equivalence of norms in finite dimensional spaces, there exists a constant $b\in\real_{++}$ such that $\|\signal{p}_n\|_1\le b$ for every $n\in\Natural$. Let $\signal{p}^\prime\in\real_{+}^N$ be an arbitrary accumulation point of $(\signal{p}_n)_{n\in\Natural}$. As a consequence of \cite[Theorem~6.7.2]{micheal07}, we can extract from $(\signal{p}_n)_{n\in\Natural}$ a  subsequence $(\signal{p}_n)_{n\in K_1}$, $K_1\subset\Natural$, converging to  $\signal{p}^\prime\in\real_+^N$. By Fact~\ref{fact.equivalence}, the definition of the scalar $b$, and monotonicity of concave mappings (Fact~\ref{fact.concave_mappings}.2), we have
 	\begin{multline}
 	\label{eq.boundcn}
 	(\forall n\in \Natural) \\ \signal{0}< T(\signal{0})\le \dfrac{1}{c_n}\signal{p}_n = T(\signal{p}_n) \le T(b \signal{1})\text{ and }\|\signal{p}_n\|_a=\bar{p}_n.
 	\end{multline}
 	Therefore, \refeq{eq.boundcn} and monotonicity of the norms $\|\cdot\|_a$ and $\|\cdot\|_1$ yield
 	\begin{multline}
 	\label{eq.bounded_cn}
 	(\forall n\in \Natural)\quad 0<\dfrac{\|\signal{p}_n\|_a}{\|T(b\signal{1})\|_a} = \dfrac{\bar{p}_n}{\|T(b\signal{1})\|_a} \le c_n \\ \le \dfrac{\|\signal{p}_n\|_1}{\|T(\signal{0})\|_1}\le \dfrac{b}{\|T(\signal{0})\|_1}.
 	\end{multline}
 	The above inequalities show that the sequence $(c_n)_{n\in K_1}$ is bounded, so we can use the Bolzano-Weierstrass theorem to conclude that there is a subsequence $(c_n)_{n\in K_2}$, $K_2\subset K_1\subset\Natural$, converging to a scalar $c^\prime$. From \refeq{eq.bounded_cn}, we also know that $c^\prime = \lim_{n\in K_2} c_n \ge \lim_{n\in K_2} \bar{p}_n/\|T(b\signal{1})\|_a = \bar{p}^\star/\|T(b\signal{1})\|_a > 0$. As a result, from the continuity of the mapping $T$, $K_2\subset K_1$, and \refeq{eq.boundcn},
 	\begin{align}
 	\label{eq.fixed_point}
 	\dfrac{1}{c^\prime}\signal{p}^\prime=\lim_{n \in K_2} \dfrac{1}{c_n}\signal{p}_n = \lim_{n\in K_2} T(\signal{p}_n)=T(\signal{p}^\prime),
 	\end{align}
 	and, from continuity of norms and \refeq{eq.boundcn},
 	\begin{align}
 	\label{eq.norm}
 	\|\signal{p}^\prime\|_a=\lim_{n\in K_2} \|\signal{p}_n\|_a = \lim_{n\in K_2} \bar{p}_n = \bar{p}^\star.
 	\end{align}	
 	By \refeq{eq.fixed_point} and \refeq{eq.norm} the tuple $(\signal{p}^\prime, c^\prime)$ is a solution to the conditional eigenvalue problem associated with the mapping $T$. Hence, by Fact~\ref{fact.equivalence}.1, $(\signal{p}^\prime, c^\prime)$ is the unique solution $(\signal{p}^\star,~c^\star)$ to Problem~\ref{problem.canonical} with the power budget $\bar{p}^\star$. As a result, we have $\signal{p}^\star = \signal{p}^\prime$, which by uniqueness of $\signal{p}^\star$ shows that  $(\signal{p}_n)_{n\in\Natural}\subset\real_{++}^N$ converges to $\signal{p}^\star$. \par

 	(ii) By Fact~\ref{fact.concave_mappings}.2 and monotonicity of $\|\cdot\|_a$, the inequalities $\|T(\signal{p})\|_a\ge \|T(\signal{0})\|_a>0$ hold for every $\signal{p}\in \real_{+}^N$. Since  $c_n=\bar{p}_n/\|{T}(\signal{p}_n)\|_a \in\real_{++}$ and  $\bar{p}^\star/\|T(\signal{p}^\star)\|_a=c^\star$ by Fact~\ref{fact.equivalence}.1, we use (i) and continuity of $T$ to obtain   $\lim_{n\to\infty}c_n=\lim_{n\to\infty}\bar{p}_n/\|T(\signal{p}_n)\|_a=\bar{p}^\star/\|T(\signal{p}^\star)\|_a=c^\star$.

 	(iii) As shown above, we have $\|T(\signal{p})\|_b\ge\|T(\signal{0})\|_b>0$ for every $\signal{p}\in\real_+^N$, so  ${E(\bar{p}_n)}={1}/{\|T(\signal{p}_n)\|_b}\in\real_{++}$ for every $n\in\Natural$ . Therefore, (i) and continuity of the mapping $T$ yield
 	\begin{align*}
 	\lim_{n\to\infty}{E(\bar{p}_n)}=\lim_{n\to\infty}\dfrac{1}{\|T(\signal{p}_n)\|_b} = \dfrac{1}{\|T(\signal{p}^\star)\|_b}={E(\bar{p}^\star)},
 	\end{align*}
 	which completes the proof.
 \end{proof}

\subsection{Bounds on the utility}

Fact~\ref{fact.equivalence}.2 shows that the utility $U$ (e.g., rates) increases by increasing the power budget available to transmitters. However, in the next proposition we verify that the possible gain in utility is limited in general.\footnote{One of the inequalities in Proposition~\ref{proposition.bounded_rate} has already appeared in \cite{renato2016maxmin}. However, that study considered only the $l_\infty$ norm and a very particular max-min utility maximization problem in load coupled networks.} The next proposition also derives a conservative lower bound for the utility function. The bounds shown below are useful for a quick evaluation of the network performance for different values of the power budget $\bar{p}$. 

\begin{proposition}
	\label{proposition.bounded_rate}
Consider the assumptions and definitions in Problem~\ref{problem.canonical}, and denote by $\signal{M}\in\real_{+}^{N\times N}$ the lower bounding matrix of the interference mapping $T:\real_{+}^N\to\real_{++}^N$. Let $U:\real_{++}\to\real_{++}$ be the utility function in Definition~\ref{def.mm_ee}. Then each of the following holds:
\begin{itemize} 
	\item[(i)] (Upper bound) 
	\begin{align*}
	(\bar{p}\in\real_{++})~U(\bar{p}) \le \dfrac{\bar{p}}{\|T(\signal{0})\|_a}
	\end{align*}
	 Furthermore, if $\rho(\signal{M})>0$, then we also have
	\begin{align*}
	(\forall\bar{p}\in\real_{++})~U(\bar{p}) < \dfrac{1}{\rho(\signal{M})}.
	\end{align*}
	
	\item[(ii)] (Combined upper bound)  Assume that $\rho(\signal{M})>0$, then
\begin{multline*}
(\forall \bar{p}\in\real_{++})~U(\bar{p}) \le \begin{cases}
\dfrac{1}{\rho(\signal{M})} & \text{if } \bar{p}\ge u \\ 
\dfrac{\bar{p}}{\|T(\signal{0})\|_a} & \text{ otherwise,}
\end{cases}
\end{multline*}
where $u:=\|T(\signal{0})\|_a/\rho(\signal{M})$.

\item[(iii)] (Lower bound) Let $\beta\in\real_{++}$ be an arbitrary scalar satisfying $\|\signal{p}\|_\infty \le \beta \|\signal{p}\|_a$ for every $\signal{p}\in\real_{++}^N$ (see Fact~\ref{fact.norm_equiv}). Then
\begin{align}
\label{eq.bounded_rate2}
(\forall \bar{p}\in\real_{++})\quad  U(\bar{p}) \ge \dfrac{\bar{p}}{\|T(\bar{p}\beta\signal{1})\|_a}.
\end{align}

\item[(iv)] (Asymptotic lower bound)  Denote by $t^{(i)}:\real_{+}^N\to\real_{++}$ the concave function given by $i$th component of the interference mapping $T$; i.e., $T(\signal{p})=:[t^{(1)}(\signal{p}),\ldots,t^{(N)}(\signal{p})]^t$ for every $\signal{p}\in\real_{++}^N$. Let $\beta\in\real_{++}$ be as in part (iii) of the proposition, and define  $\signal{p}_\infty:=[t^{(1)}_\infty(\signal{1}),\ldots,t^{(N)}_\infty(\signal{1})]$, where $t^{(i)}_\infty$ is the asymptotic function of $t^{(i)}$ for $i\in\{1,\ldots,N\}$ (see Definition~\ref{definition.recession}). Then we have

\begin{align}
\label{eq.rate_unlimited}
c_\infty:=\lim_{\bar{p}\to\infty}U(\bar{p}) \ge \begin{cases} \dfrac{1}{\beta \|\signal{p}_\infty\|_a} & \quad \mathrm{if}\quad \|\signal{p}_\infty\|_a>0 \\
\infty & \mathrm{otherwise.}
\end{cases}
\end{align}
\end{itemize}
\end{proposition}
\begin{proof}
	In this proof, we use the standard notation  $P:\real_{++}\to\real_{++}^N$ for the power function in Definition~\ref{def.mm_ee}.

	(i)  Define by $\|\cdot\|_{\bar{p}}$ the monotone norm $\|\signal{p}\|_{\bar{p}} := (1/\bar{p}) \|\signal{p}\|_a$ for a given power budget $\bar{p}\in\real_{++}$. By Fact \ref{fact.equivalence}.1, we have 
	\begin{multline}
	\label{eq.loc_ineq1}
	(\forall \bar{p}\in\real_{++}) \\
	\dfrac{1}{U(\bar{p})} P(\bar{p}) = T(P(\bar{p})) \text{ and } \|P(\bar{p})\|_{\bar{p}} = \|P(\bar{p})\|_a/{\bar{p}} = 1
	\end{multline}
	
	Now, monotonicity of both the norm $\|\cdot\|_a$ and the mapping $T$ (Fact~\ref{fact.concave_mappings}.2) yields
	\begin{multline*}
	(\forall\bar{p}\in\real_{++})  \frac{\bar{p}}{U(\bar{p})} = \frac{1}{U(\bar{p})} \|P(\bar{p})\|_a \\  = \|T(P(\bar{p}))\|_a \ge  \|T(\signal{0})\|_a > 0.
	\end{multline*}
	Hence $U(\bar{p}) \le \bar{p}/\|T(\signal{0})\|_a$. Now assume that $\rho(\signal{M})>0$, in which case $U(\bar{p}) < {1}/{\rho(\signal{M})}$ immediately follows from \refeq{eq.loc_ineq1} and Fact~\ref{fact.eigbound}.
	
	(ii) Immediate from the two inequalities in (i).

	(iii) From the definition of the scalar $\beta$ and Fact~\ref{fact.equivalence}.1, we deduce $\|P(\bar{p})\|_\infty\le \beta \|P(\bar{p})\|_a=\beta \bar{p}$ for all $\bar{p}\in\real_{++}$. Therefore, $(\forall \bar{p}\in\real_{++})~~ P(\bar{p})\le \|P(\bar{p})\|_\infty~\signal{1}\le \beta\|P(\bar{p})\|_a\signal{1}=\bar{p}\beta\signal{1}$. By using this relation together with Fact~\ref{fact.equivalence}, the monotonicity of $\|\cdot\|_a$, and monotonicity of $T$ (Fact~\ref{fact.concave_mappings}.2) we obtain
		\begin{multline*}
	(\forall \bar{p}\in\real_{++})\quad	0<\frac{\bar{p}}{U(\bar{p})} = \frac{\|P(\bar{p})\|_a}{U(\bar{p})}  \\ = \|T(P(\bar{p}))\|_a \le \|T(\bar{p}\beta\signal{1})\|_a, 
		\end{multline*}
	which implies \refeq{eq.bounded_rate2} and completes the proof of (ii).
	
	(iv) From (iii) we know that $0<1/U(\bar{p})\le (1/\bar{p})\|T(\bar{p}\beta\signal{1})\|_a$ for every $\bar{p}\in\real_{++}$. By taking the limit as $\bar{p}\to\infty$ and by considering \refeq{eq.recession_func} we deduce
	\begin{align}
	\label{eq.asymp}
	0\le \lim_{\bar{p}\to\infty}\dfrac{1}{U(\bar{p})}\le \|[t^{(1)}_\infty(\beta\signal{1}),\ldots,t^{(N)}_\infty(\beta\signal{1})]\|_a=\beta\|\signal{p}_\infty\|_a,
	\end{align}
	where the last equality comes from the fact that asymptotic functions are positively homogeneous (see Definition~\ref{definition.recession}). The result in \refeq{eq.asymp} implies the inequality in \refeq{eq.rate_unlimited}.   
\end{proof}

 A possible application of Proposition~\ref{proposition.bounded_rate} is to identify noise limited networks and interference limited networks  with a precise mathematical statement. More specifically, assume that we have a max-min fairness problem that can be posed in the canonical form in \refeq{eq.canonical}. For this problem, if $c_\infty$ in Proposition~\ref{proposition.bounded_rate}(iii) is infinity, then any utility (e.g., rate) is achievable if the power budget is sufficiently large. In contrast, if the lower bounding matrix $\signal{M}$ of the interference mapping $T$ satisfies $\rho(\signal{M})>0$, then Proposition~\ref{proposition.bounded_rate}(i) shows that the utility cannot grow unboundedly. In other words, interference cannot be overcome by increasing the power budget. As expected, the assumption $\rho(\signal{M})>0$ is typically satisfied in network optimization problems with links directly or indirectly coupled by interference; i.e., changes in the power of any link has influence on the interference received by any other link in the network.  If links are not coupled, then we usually have $c_\infty=\infty$ (which also implies $\rho(\signal{M})=0$), and in this case the utility can be as large as desired by using a sufficiently large power budget.

\subsection{Bounds on the energy efficiency}

We now proceed to derive bounds on the energy efficiency as a function of the power budget $\bar{p}$.

\begin{proposition}
	\label{proposition.monotone_ee}
Let $E:\real_{++}\to\real_{++}$ be the $\|\cdot\|_b$-energy efficiency function in Definition~\ref{def.mm_ee}. Then we have the following.
\begin{itemize}

\item[(i)] (Lower bound) Let $\alpha_1\in\real_++$ be an arbitrary scalar satisfying $\alpha_1 \|\signal{p}\|_b\le \|\signal{p}\|_a $ for every $\signal{p}\in\real^N$, and $\beta\in\real_{++}$ an arbitrary scalar satisfying $\|\signal{p}\|_\infty \le \beta \|\signal{p}\|_a$ for every $\signal{p}\in\real_{++}^N$ (see Fact~\ref{fact.norm_equiv}). Then 
\begin{align}
\label{eq.lb_ee}
(\forall {\bar{p}}\in\real_{++}) \quad \dfrac{\alpha_1}{ \|T(\bar{p}\beta\signal{1})\|_a} \le E(\bar{p}).
\end{align}

\item[(ii)] (Upper bound) Assume that $\rho(\signal{M})>0$, where $\signal{M}$ is the lower bounding matrix of the interference mapping $T:\real_{+}^N\to\real_{++}^N$. Then
\begin{align}
\label{eq.ub_ee}
(\forall {\bar{p}}\in\real_{++}) E(\bar{p}) < \dfrac{\alpha_2}{\rho(\signal{M}){\bar{p}}},
\end{align}
where  $\alpha_2\in \real_{++}$ is an arbitrary scalar satisfying $\|\signal{p}\|_a \le \alpha_2 \|\signal{p}\|_b$ for every $\signal{p}\in\real^N$.

\item[(iii)] (Maximum energy efficiency)
\begin{align}
\label{eq.blimit}
\lim_{\bar{p}\to 0^+} E(\bar{p})=\sup_{\bar{p}>0} E(\bar{p}) = \dfrac{1}{\|T(\signal{0})\|_b}.
\end{align}

\item[(iv)] (Combined upper bound) Let $\alpha_2\in\real_{++}$ be as in (ii). Then 
\begin{multline*}
(\forall \bar{p}\in\real_{++})~E(\bar{p}) \le \begin{cases}
\dfrac{\alpha_2}{\rho(\signal{M})\bar{p}} & \text{if } \bar{p}\ge k \\ 
\dfrac{1}{\|T(\signal{0})\|_b} & \text{ otherwise,}
\end{cases}
\end{multline*}
where $k:=\|T(\signal{0})\|_b/\rho(\signal{M})$.

\end{itemize}
\end{proposition}
\begin{proof}\par 
 (i) Divide both sides of \refeq{eq.bounded_rate2} by $\|P(\bar{p})\|_b>0$ and use $\bar{p}=\|P(\bar{p})\|_a \ge \alpha_1 \|P(\bar{p})\|_b$ to obtain the sought lower bound
 \begin{align*}
 E(\bar{p})=\dfrac{U(\bar{p})}{\|P(\bar{p})\|_b} \ge \dfrac{\bar{p}}{\|T(\bar{p}\beta\signal{1})\|_a ~\|P(\bar{p})\|_b} \ge \dfrac{\alpha_1}{ \|T(\bar{p}\beta\signal{1})\|_a}.
 \end{align*} \par 
 
(ii) 	Recall that $\|P(\bar{p})\|_a=\bar{p}$ by Fact~\ref{fact.equivalence}.1 for every $\bar{p}\in\real_{++}$. Since $\rho(\signal{M})>0$ by assumption, it follows from  Proposition~\ref{proposition.bounded_rate}(i) that $(\forall \bar{p}\in\real_{++})~U(\bar{p}) < 1/\rho(\signal{M})$.  Therefore, 
\begin{multline*}
(\forall\bar{p}\in\real_{++})~
E(\bar{p})=\dfrac{U(\bar{p})}{\|P(\bar{p})\|_b} < \dfrac{1}{\rho(\signal{M}) \|P(\bar{p})\|_b} \\ \le \dfrac{\alpha_2}{\rho(\signal{M}) \|P(\bar{p})\|_a}=\dfrac{\alpha_2}{\rho(\signal{M}) \bar{p}},
\end{multline*}
which concludes the proof of the upper bound.

 \par

(iii) Existence of the limit $\lim_{\bar{p}\to 0^+} E(\bar{p})$ and the equality $\lim_{\bar{p}\to 0^+} E(\bar{p})=\sup_{\bar{p}>0} E(\bar{p})$ is immediate from non-negativity of $E$ and Lemma~\ref{lemma.ee_prop}. To obtain the value of the limit, denote by $U$ and $P$ the utility and power functions, respectively. Let $(\bar{p}_n)_{n\in\Natural}\subset\real_{++}$ be an arbitrary sequence satisfying $\lim_{n\to\infty} \bar{p}_n = 0$, and let the sequence of tuples $\left((\signal{p}_n,c_n)\right)_{n\in\Natural}\subset\real^N_{++}\times\real_{++}$ be given by $(\signal{p}_n,c_n):=(P(\bar{p}_n), U(\bar{p}_n))$ for every $n\in\Natural$.  Since $\lim_{n\to\infty} \|\signal{p}_n\|_b = \lim_{n\to\infty} \bar{p}_n = 0$, we have $\lim_{n\to\infty} \signal{p}_n = \signal{0}$. Recalling that $T:\real_{+}^N\to\real_{++}^N$ is continuous by assumption and that norms are continuous, we obtain:
\begin{align*}
\lim_{n\to \infty} \dfrac{1}{E(\bar{p}_n)}=\lim_{n\to \infty}\|T(\signal{p}_n)\|_b = \|T(\signal{0})\|_b>0,
\end{align*}
which implies \refeq{eq.blimit}.

(iv) Immediate from (ii) and (iii).
\end{proof}

 \subsection{Asymptotic behavior of the utility and the energy efficiency functions}

With the results obtained in the previous subsections, we have all the ingredients necessary to study the behavior of the utility and energy efficiency functions as $\bar{p}\to\infty$. We start by studying the utility function.

\begin{Cor}
	Assume that the lower bounding matrix $\signal{M}$ of the concave mapping $T$ in Problem~\ref{problem.canonical} satisfies $\rho(\signal{M})>0$. Then $U(\bar{p})\in \Theta(1)$ as $\bar{p}\to\infty$.
\end{Cor}
\begin{proof}
	With the assumption $\rho(\signal{M})>0$, Proposition~\ref{proposition.bounded_rate}(i) shows that the utility function $U$ is bounded above by $1/\rho(\signal{M})$. Furthermore, Fact~\ref{fact.equivalence}.2 shows that $U$ is increasing, hence the limit $\lim_{\bar{p}\to\infty} U(\bar{p})>0$ exists, which implies $U(\bar{p})\in \Theta(1)$ as $\bar{p}\to\infty$. 
\end{proof}

 The next example shows that known results in the literature \cite[Ch. 5]{slawomir09} are particular cases of the upper bound in Proposition~\ref{proposition.bounded_rate}(ii). In addition, this example proves that the upper bound $1/\rho(\signal{M})$ for the utility (assuming $\rho(\signal{M})>0$)  is asymptotically sharp for an important class of interference mappings that are common in max-min fairness problems in wireless networks \cite{slawomir09}.
 
 \begin{example}
 	\label{example.tight}
 	Let $\signal{\sigma}\in\real_{++}^N$ be arbitrary and assume that $\signal{M}\in\real_{+}^{N\times N}$ is a matrix satisfying $\rho(\signal{M})>0$. Now consider Problem \ref{problem.canonical} with the affine interference mapping $T:\real_+^N\to\real_{++}^N:\signal{p}\mapsto\signal{M}\signal{p}+\signal{\sigma}$, and note that the lower bounding matrix of $T$ is the matrix $\signal{M}$. The results in \cite[Theorems A.16 and A.51]{slawomir09} show that, for any $\epsilon\in ~]0,1[$, the vector $\signal{p}_\epsilon={c}_\epsilon(\signal{I}-{c}_\epsilon \signal{M})^{-1} \signal{\sigma}$ is strictly positive, where we define $c_\epsilon:=\epsilon/\rho(\signal{M})$. Therefore, given the power budget $\bar{p}_\epsilon = \|{c}_\epsilon(\signal{I}-{c}_\epsilon \signal{M})^{-1}\signal{\sigma}\|_a=\|\signal{p}_\epsilon\|_a$, we have $\signal{p}_\epsilon\in\mathrm{Fix}(c_\epsilon T)\neq\emptyset$. Fact~\ref{fact.equivalence}.1 now shows that the tuple $(\signal{p}_\epsilon, {c}_\epsilon)$ is the solution to Problem~\ref{problem.canonical} with the power budget $\bar{p}_\epsilon=(\epsilon/\rho(\signal{M})) ~ \|(\signal{I}-(\epsilon/\rho(\signal{M}))~\signal{M})^{-1}\signal{\sigma}\|_a$. This result proves that any utility strictly smaller than $1/\rho(\signal{M})$ can be achieved for the affine interference mappings considered in this example, provided that the power budget is large enough. More precisely, Fact~\ref{fact.equivalence}.2 shows that 
 	\begin{align}
 	\label{eq.bd_asymp}
 	(\forall\epsilon\in~]0,1[)(\forall  \bar{p}\ge \bar{p}_{{\epsilon}}) \quad U(\bar{p}) \ge c_\epsilon=\epsilon/\rho(\signal{M}),
 	\end{align}
 	where $U$ is the utility function in Definition~\ref{def.mm_ee}. 
 \end{example}

Scaling properties of the energy efficiency function are shown in the next corollary.

\begin{Cor} 
	\label{cor.asymp_e}
	Let $\signal{M}$ be the lower bounding matrix of the interference mapping $T$ in Problem~\ref{problem.canonical}. The $\|\cdot\|_b$-energy efficiency function has the following properties:
	\begin{itemize}
		
	\item[(i)] The limit $\lim_{\bar{p}\to\infty} E(\bar{p})$ exists and $\inf_{\bar{p}>0} E(\bar{p}) = \lim_{\bar{p}\to\infty} E(\bar{p})<\infty $.  Furthermore, if $\rho(\signal{M})>0$, then $\lim_{\bar{p}\to\infty} E(\bar{p})=0$.
		\item[(ii)] Assume that $\rho(\signal{M})>0$, then		
		\begin{align*}
		(\exists k_1\in\real_{++})(\exists k_2\in\real_{++})(\forall \bar{p}\in\real_{++})
		 \dfrac{1}{k_1 \bar{p}+ k_2} \le E(\bar{p}).
		\end{align*}	
		\item[(iii)] With the assumption in (ii), we also have  $E(\bar{p})\in \Theta(1/\bar{p})$ as $\bar{p}\to\infty$.
	\end{itemize}	
\end{Cor}
\begin{proof}
	(i) Immediate from Lemma~\ref{lemma.ee_prop}, non-negativity of $E$, and Proposition~\ref{proposition.monotone_ee}(ii). \\ 
	
	(ii) Proposition~\ref{proposition.affine_bound}, Proposition~\ref{proposition.monotone_ee}(i), and the equivalence of norms show that there exists $(k_1, k_2)\in\real_{++}\times\real_{++}$ such that $0<1/E(\bar{p}) \le \|T(\bar{p} \beta\signal{1})\|_b \le k_1 \bar{p}+ k_2$ for every $\bar{p}\in\real_{++}$, and the desired inequality follows. 
	
	(iii) By the inequality obtained in (ii) and Proposition~\ref{proposition.monotone_ee}(ii), we have:
	\begin{multline*}
	(\exists k_1\in\real_{++})(\exists k_2\in\real_{++})(\exists k_3\in\real_{++})(\forall \bar{p}\ge 1)  \\
	\quad \dfrac{1}{ (k_1+k_2) \bar{p}}  \le \dfrac{1}{ \|T(\bar{p}\beta\signal{1})\|_b}  \le E(\bar{p}) \le  
	\dfrac{k_3}{\rho(\signal{M})\bar{p}},
	\end{multline*}
	and hence $E(\bar{p})\in\Theta(1/\bar{p})$ as $\bar{p}\to\infty$.
\end{proof}

We can also show that the upper bound in Proposition~\ref{proposition.monotone_ee}(ii) is as sharp as possible, in the sense that, if we consider all positive concave mappings $T:\real_+^N\to\real_{++}^N$, there are mappings for which the upper bound is achieved asymptotically as the power budget diverges to infinity:

\begin{example}
	\label{example.ee_affine}
	Consider the model and definitions in Example~\ref{example.tight}. Recalling that $\|P(\bar{p})\|_a=\bar{p}$ for every $\bar{p}\in\real_{++}$ by Fact~\ref{fact.equivalence}.1, where $P:\real_{++}\to\real_{++}^N$ is the power function in Definition~\ref{def.mm_ee}, we deduce from \refeq{eq.bd_asymp}: 
	\begin{align*}
	(\forall\epsilon\in~]0,1[)(\forall  \bar{p}\ge \bar{p}_{{\epsilon}})~E(\bar{p})=\dfrac{U(\bar{p})}{\|P(\bar{p})\|_a} = \dfrac{U(\bar{p})}{\bar{p}} \ge \dfrac{\epsilon}{\rho(\signal{M}) \bar{p}},
	\end{align*}
	where $E$ is the $\|\cdot\|_a$-energy efficiency function in Definition~\ref{def.mm_ee}. In words, the above proves that the $\|\cdot\|_a$-energy efficiency function can be made arbitrarily close to the upper bound in Proposition~\ref{proposition.monotone_ee}(ii) with the choice $\alpha_1=\alpha_2=1$ for a sufficiently large power budget $\bar{p}$. 
\end{example}

\subsection{Discussion - Power regimes in network utility maximization problems}
\label{sect.illustration}

The upper bound in Proposition~\ref{proposition.bounded_rate}(ii) motivates the definition of two power regimes: a low power regime ($\bar{p}\le \|T(\signal{0})\|_a/\rho(\signal{M})$), where the linear bound in Proposition~\ref{proposition.bounded_rate}(ii) is effective, and a high power regime ($\bar{p} > \|T(\signal{0})\|_a/\rho(\signal{M})$), where the constant bound in Proposition~\ref{proposition.bounded_rate}(ii) is effective. The transition between these  two regimes happens at the transient point $\bar{p} = u := \|T(\signal{0})\|_a/\rho(\signal{M})$. For convenience, let us focus for the moment on the $\|\cdot\|_a$-energy efficiency function in Definition~\ref{def.mm_ee}, which we denote by $E_a$. In the low power regime, by \refeq{eq.blimit}, we verify that $E(\bar{p})\in\Theta(1)$ as $\bar{p}\to 0^+$, hence the decay in energy efficiency as $\bar{p}$ increases can be small. In contrast, the utility $U$ increases at best linearly in this power regime (see Proposition~\ref{proposition.bounded_rate}(ii)). These observations show that we should transmit with low power and low utility (e.g., rates) if high transmit energy efficiency is desired, and we emphasize that we have proved this expected result by using a very general model that unifies, within a single framework, the behavior of a large array of transmission technologies in wireless networks.

In the high power regime, the energy efficiency eventually decays quickly as the power budget $\bar{p}$ diverges to infinity because $E_a(\bar{p})\in\Theta(1/\bar{p})$ as $\bar{p}\to\infty$, while gains in utility eventually saturate because of the uniform bound $(\forall\bar{p}\in\real_{++})~U(\bar{p})\le 1/\rho(\signal{M})$. The above observations are illustrated in Fig.~\ref{fig.example}. Note that the $\|\cdot\|_a$-energy efficiency function $E_a$ is continuous, converges to the upper bound $1/\|T(\signal{0})\|_a$ as the power budget $\bar{p}$ decreases to zero, and decreases within the hatched areas corresponding to the power regimes. Furthermore, the utility function $U$ is also continuous, converges to zero as the power budget $\bar{p}$ decreases to zero, and increases within the hatched areas corresponding to the power regimes.

Interestingly, the above properties of the energy efficiency and the utility  hold for all network utility maximization problems that can be posed in the canonical form in \refeq{eq.canonical}. Equipping the network with advanced communication technologies such as massive multi-antenna transmission schemes, transmitter and receiver beamforming techniques, intelligent user-base station assignment mechanisms, or a combination of all these technologies can only change the position of the transient point and move upwards the bounds on the energy efficiency and the utility. However, if $\rho(\signal{M})>0$,  these two power regimes are always present.

\begin{figure}
	\centering
	\begin{subfigure}[b]{0.9\linewidth}
		\centering
		\includegraphics[width=\linewidth,keepaspectratio,bb=30mm 10mm 295mm 183mm,clip]{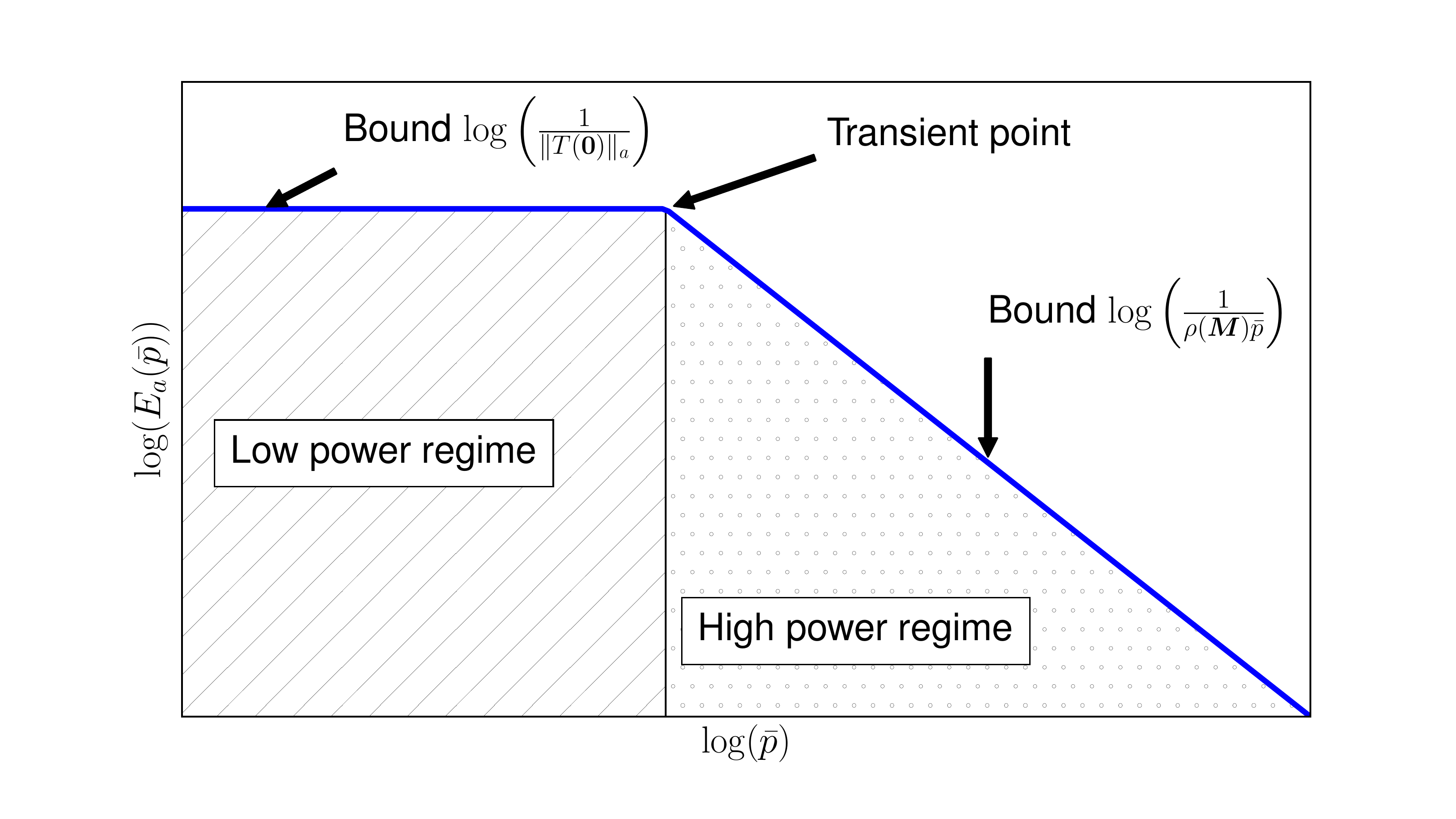}
		\caption{~}
		\label{fig.ee}
	\end{subfigure}
	\begin{subfigure}[b]{0.9\linewidth}
		\centering
		\includegraphics[width=\linewidth,keepaspectratio,bb=30mm 10mm 295mm 183mm,clip]{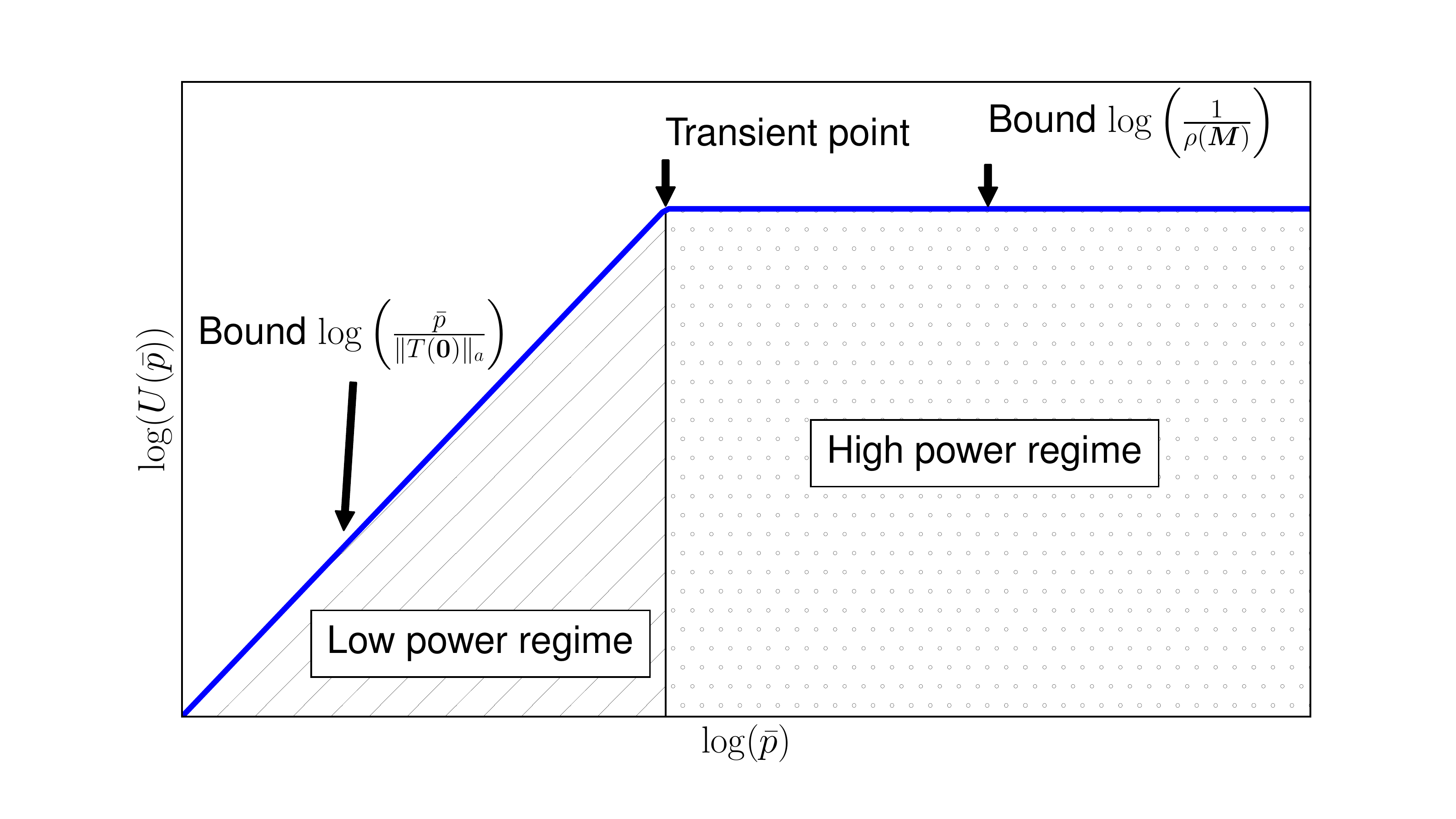}
		\caption{~}
		\label{fig.power_increase_load}
	\end{subfigure} 
	\caption{Power regimes for the solution to Problem~\ref{problem.canonical} as a function of the power budget. The lower bounding matrix $\signal{M}$ associated with the mapping $T$ is assumed to satisfy $\rho({\signal{M}})>0$. (a) Energy efficiency. (b) Utility. }\label{fig.example}
\end{figure}

\section{Exemplary application}
\label{sect.examples}

We now illustrate the results obtained in the previous section with a concrete application; namely, a novel joint power control and base station allocation problem for weighted rate maximization.

\subsection{System model and the proposed algorithm}
\label{sect.joint_pc_ba}

Denote by $\setm$ and by $\setn:=\{1,\ldots,N\}$ the set of base stations and the set of users, respectively, in a network. The pathloss between base station $i\in\setm$ and user $j\in\setn$ is indicated by the variable $g_{i,j}\in\real_{++}$. Let $\signal{p}=[p_1,\ldots,p_N]\in\real_{++}^N$ be the uplink power vector, where $p_j\in\real_{++}$ is the uplink transmit power of user $j\in\setn$. The SINR of user $j\in\setn$ if connected to base station $i\in\setm$ for a given power allocation $\signal{p}=[p_1,\dots,p_N]\in\real_{++}^N$ is given by 
\begin{align*}
s_j:\real_{++}^N\times\setm\to\real_+:(\signal{p},i)\mapsto\dfrac{p_j g_{i,j}}{\sum_{k\in\setn\setminus\{j\}} p_k g_{i,k} + \sigma^2_i}.
\end{align*}

We assume that users connect to their best serving base stations and that each user is equipped with a single-user decoder that treats interference as noise. With these assumptions, the achievable rate of user $j\in\setn$ for a given power allocation $\signal{p}\in\real_{++}^N$ under Gaussian noise is commonly approximated by $\max_{i\in\setm} B\log_2(1+s_{j}(\signal{p},i))$, where $B$ is the system bandwidth. 

The utility maximization problem we solve in this section, which we refer to as the weighted rate allocation problem, is formally stated as:
\begin{problem} (The weighted rate allocation problem)
	\label{problem.pro_fairness}
	\begin{align}
	\label{eq.pro_fairness}
	\begin{array}{lll}
	\mathrm{maximize}_{\signal{p}, c} & c \\
	\mathrm{subject~to} & c w_j = \max_{i\in\setm}  B\log_2(1+s_{j}(\signal{p},i))~(\forall j\in\setn)\\
	& \|\signal{p}\|_\infty \le \bar{p} \\
	& \signal{p}\in\real_{++}^N, c\in\real_{++},
	\end{array}
	\end{align}
\end{problem}
where $w_j\in\real_{++}$ is the weight or priority assigned to the rate $c\omega_j$ of user $j\in\setn$. 

To gain intuition on the weights $(w_j)_{j\in\setn}$, we now relate two particular choices to traditional network utility maximization problems. The simplest choice is the use of uniform weights $w_k=w_j=1$ for every $(k,j)\in\setn\times\setn$, in which case the solution to Problem~\ref{problem.pro_fairness} can be shown to maximize the minimum observed rate in the network (max-min fairness). As a second example of a weighting scheme related to that used later in the simulations, for a fixed power $\bar{u}\in\real_{++}$, we can use $w_j=\max_{i\in\setm} B\log_2(1+ \bar{u}g_{i,j}/\sigma^2_i)$ for every $j\in\setn$. With this choice, $w_j$ is simply the best rate that each user $j$ can achieve if alone in the system and transmitting at full power $\bar{u}$. In particular, if $\bar{u}=\bar{p}$, then the solution $(\signal{p}^\star, c^\star)$ to Problem~\ref{problem.pro_fairness}  allocates to each user $j\in\setn$ the fraction $c^\star\in~]0,1]$ of its best individual achievable rate $w_j$. The number $c^\star$ is the largest fraction common to all users that the network can support.

We now proceed to state Problem~\ref{problem.pro_fairness} in the canonical form in \refeq{eq.canonical}. For any $\signal{p}\in\real^N_{++}$, $c\in\real_{++}$, and $j\in\setn$ satisfying the first constraint in \refeq{eq.pro_fairness}, we have
\begin{multline*}
c w_j = \max_{i\in\setm}  B\log_2(1+s_{j}(\signal{p},i)) \\ \Leftrightarrow \dfrac{1}{c}=\min_{i\in\setm}\dfrac{w_j}{B\log_2(1+s_{j}(\signal{p},i))}  
 \\ \Leftrightarrow \dfrac{1}{c} p_j=\min_{i\in\setm}\dfrac{w_j p_j}{B\log_2(1+s_{j}(\signal{p},i))},
\end{multline*}
which shows that the first constraint in Problem \ref{problem.pro_fairness} can be equivalently written as
\begin{align}
\label{eq.opendomain}
\signal{p}\in\mathrm{Fix}(c\tilde{T}),
\end{align}
where $\tilde{T}:\real_{++}^N\to\real_{++}^N:\signal{p}\mapsto [\tilde{t}_1(\signal{p}),\ldots,\tilde{t}_N(\signal{p})]$ and
\begin{align*}
\begin{array}{rcl}
\tilde{t}_j:\real_{++}^N&\to&\real_{++}^N\\
\signal{p}&\mapsto&\min_{i\in\setm}\dfrac{w_j p_j}{B\log_2(1+s_{j}(\signal{p},i))}
\end{array}
\end{align*}
for every $j\in\setn$.

For fixed $i\in\setm$, we know by \cite{cai2012optimal,renato14SIP}\cite[Lemma 2]{renato2016power} that ${w_j p_j}/{(B\log_2(1+s_{j}(\signal{p},i)))}$ as a function of $\signal{p}\in\real^N_{++}$ is concave\footnote{After the submission of the first version of the manuscript, one of the reviewers pointed out that \cite{cai2012optimal} is possibly the first study to show that this function is concave. The studies in \cite{renato14SIP,renato2016power} show all formal details of the continuous extension of this function to the boundary of the domain. This continuous extension is crucial to the bounds we derive because the mapping $T$ in Problem
	\ref{problem.canonical} is assumed to be continuous. It is also worth mentioning that the function $\log(1+\text{SINR})$ is known to be log-log concave \cite{papa2008}, a property that has already been exploited for many years to solve utility maximization problems without using common simplifications such as those based on the high SINR assumption considered in \cite{chiang2005}, for example.} for every $j\in\setn$, hence $\tilde{t}_j$ is also concave for every $j\in\setn$ because it is the minimum of concave functions \cite[Proposition~8.14]{baus11}. Therefore, by replacing the rate constraint in Problem~\ref{problem.pro_fairness} by the constraint in \refeq{eq.opendomain}, we have Problem~\ref{problem.pro_fairness} in the canonical form in \refeq{eq.canonical}, except that the domain of the mapping $\tilde{T}$ must be extended to include the boundary $\real_{+}^N\setminus\real_{++}^N$. By \cite[Theorem~10.3]{rock70}, we know that each function  $\tilde{t}_j$, $j\in\setn$, has one and only one continuous extension to $\real_{+}^N$, and by following similar arguments to those used in the proof of \cite[Lemma~3]{renato2016power}, the continuous extension of the mapping $\tilde{T}$ is given by  
\begin{align}
\label{eq.cont_ext}
{T}:\real_{+}^N\to\real_{++}^N:\signal{p}\mapsto [{t}_1(\signal{p}),\ldots,{t}_N(\signal{p})],
\end{align}
where
\begin{align*}
\begin{array}{rcl}
t_j:\real_{+}^N&\to&\real_{++}\\
\signal{p}&\mapsto&\begin{cases}\min_{i\in\setm}\dfrac{w_j p_j}{B\log_2(1+s_{j}(\signal{p},i))}~~  \text{if }p_j\neq 0 \\
\min_{i\in\setm}\dfrac{w_j\ln 2}{Bg_{i,j}}\left(\sum_{k\in\setn\backslash\{j\}}p_k g_{i,k}+\sigma^2_i\right) \\ \qquad\qquad\qquad \text{otherwise,}
\end{cases}
\end{array}
\end{align*}
for every $j\in\setn$.

The above discussion shows that the weighted rate allocation problem is equivalent to the following problem in the canonical form:
\begin{problem}
	\label{problem.pro_fairness_canonical}
	\begin{align*}
	\begin{array}{lll}
	\mathrm{maximize}_{\signal{p}, c} & c \\
	\mathrm{subject~to} & \signal{p}\in\mathrm{Fix}(cT) \\
	& \|\signal{p}\|_\infty \le \bar{p} \\
	& \signal{p}\in\real_{+}^N,~ c\in\real_{++},
	\end{array}
	\end{align*}
\end{problem}
where in the application under consideration the mapping $T:\real_{+}^N\to\real_{++}^N$ is given by \refeq{eq.cont_ext}, and we assume that $\bar{p}>0$. 
In light of Fact~\ref{fact.equivalence}.1, we can now use the simple iterative scheme in \refeq{eq.krause_iter} with the monotone norm $\|\signal{p}\|:=\|\signal{p}\|_a/\bar{p}$ to solve Problem~\ref{problem.pro_fairness}. Furthermore, all bounds derived in Sect.~\ref{sect.ee} are available because the continuous mapping $T$ is defined at $\signal{0}$.  Once a solution $(\signal{p}^\star,c^\star)$ is obtained, we recover an optimal user-base station assignment by connecting each user $j\in\setn$ to any base station $i\in\argmax_{l\in\setm}s_j(\signal{p}^\star,l)$ providing the best SINR $\max_{l\in\setm}s_j(\signal{p}^\star,l)$.

Note that the component of the $k$th column and $j$th row of the lower bounding matrix $\signal{M}\in\real_{+}^{N\times N}$ of the interference mapping $T$ in \refeq{eq.cont_ext} is given by
\begin{align*}
[\signal{M}]_{j,k}=\begin{cases}
0, &\mathrm{if}\quad k=j \\
\min_{i\in\setm}\dfrac{w_j g_{i,k}\mathrm{ln}(2)}{Bg_{i,j}} & \mathrm{otherwise.}
\end{cases}
\end{align*}

With the assumption $g_{i,k}\in\real_{++}$ for every $i\in\setm$ and $j\in\setn$, we observe that $\signal{M}$ is irreducible \cite[Sect.~A.4.1]{slawomir09} because only the diagonal terms of the non-negative matrix $\signal{M}$ are zero. Therefore, we conclude that $\rho(\signal{M})>0$ because the spectral radius of an arbitrary irreducible matrix is positive \cite[p. 673]{meyer2000}. The practical implication of this fact is that, by Proposition~\ref{proposition.bounded_rate}, the network considered here is interference limited because the rates cannot grow unboundedly by increasing the power budget. 

The bounds in Proposition~\ref{proposition.bounded_rate} and Proposition~\ref{proposition.monotone_ee} can also provide us with intuition on optimal assignment strategies in the low power regime. For example, Proposition~\ref{proposition.monotone_ee} shows that the energy efficiency function, which is continuous, converges to $1/\|T(\signal{0})\|_a$ as the power budget decreases to zero. The $j$th component of the vector $T(\signal{0})$ is $t_j(\signal{0})=\min_{i\in\setm}w_j \sigma_i^2 \ln 2/(B g_{i,j})$ by definition, which shows that, as the power budget decreases to zero, each user $j\in\setn$ selects the base station $i\in\setm$ with the smallest ratio $\sigma_i^2/g_{i,j}$. If the noise power $\sigma_i^2$ is the same at every base station $i\in\setm$, then we obtain the intuitive result that each user selects the base station with the best propagation condition, which is the expected result in a regime where noise dominates interference. Note that this result is valid for any choice of positive weights $(\omega_j)_{j\in\setn}$.

\subsection{Simulations}
\label{sect.joint_pc_ba_simulations}
The simulations in this section are based on the ``dense urban information society'' scenario provided by the METIS project \cite[Sect. 4.2]{metis_D_6_1}\cite{metis_raytrace}, which is depicted in Fig.~\ref{fig:madrid}. For simplicity, we only use the microcells of that scenario. The pathloss between users and microcells are obtained from the lookup tables available at \cite{metis_raytrace}. For the simulation, we place 30 users uniformly at random on the streets or the park within the box region $10\mathrm{m} \le x \le 377\mathrm{m}$ and  $10\mathrm{m} \le y \le 542\mathrm{m}$ of the network in Fig.~\ref{fig:madrid}. Users are free to connect to any microcell in the whole region. The noise power spectral density at every base station is $-145$~dBm/Hz, and the total system bandwidth is $10$~MHz. 

\begin{figure}
	\begin{center}
		\includegraphics[width=0.6\linewidth, angle=-90]{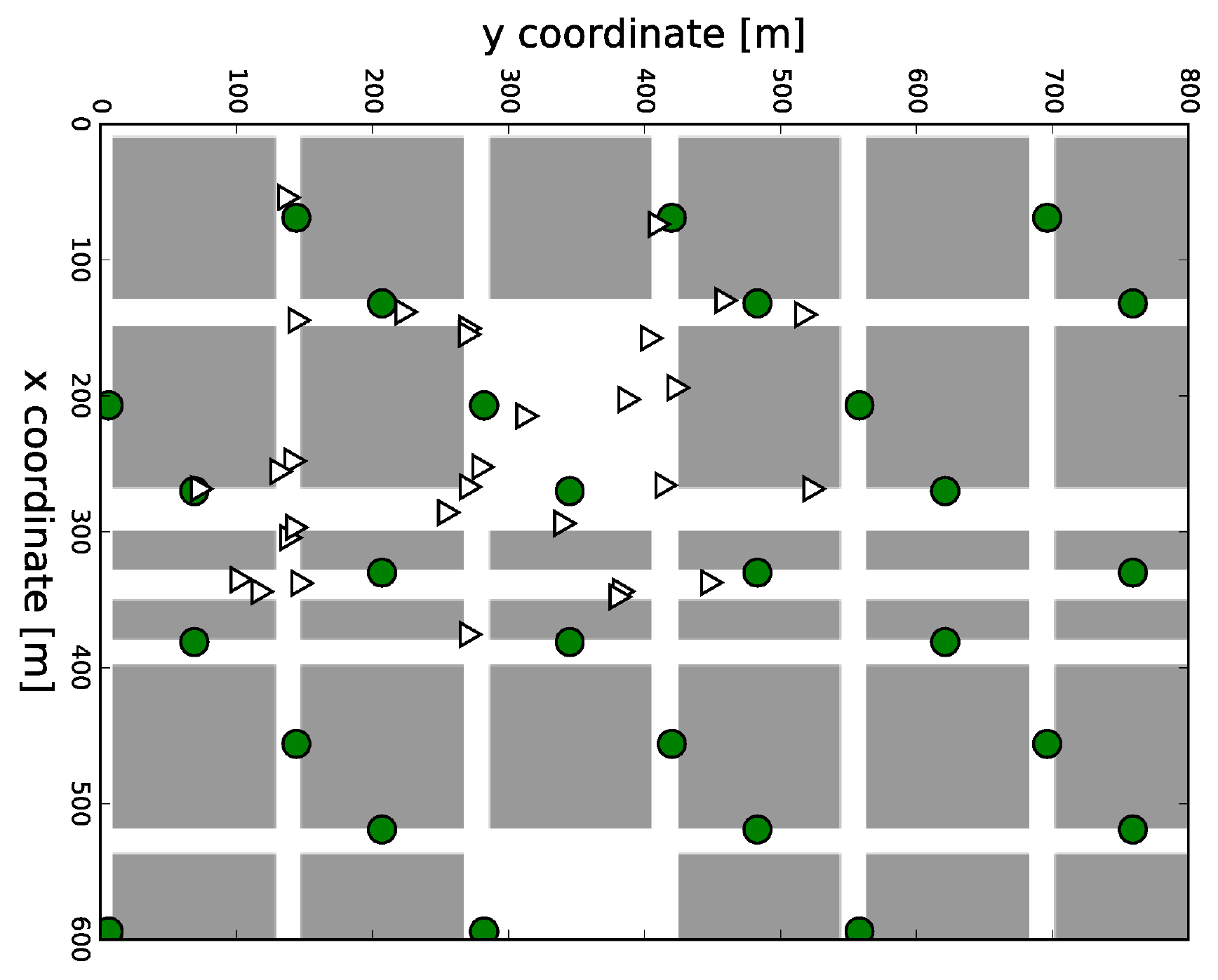}
		\caption{\small Madrid scenario of the METIS project. Buildings are represented by dark gray squares and rectangles. Areas in white color indicate the presence of streets or parks. Circles and triangles indicate the location of the microcells and users, respectively.}
		\label{fig:madrid}
	\end{center}
	\vspace{-0.5cm}
\end{figure}

The user priorities $\signal{w}=[w_1,\ldots,w_N]$ to build the concave mapping $T$ in \refeq{eq.cont_ext} are assigned as follows. We first compute the interference-free rate $b_j\in\real_{++}$ of each user $j\in\setn$ when transmitting at $1$~W; i.e., $b_j=\max_{i\in\setm} B\log_2(1+ {1}g_{i,j}/\sigma^2_i)$. We then use the weight vector $\signal{w}=\signal{b}/\|\signal{b}\|_\infty$, where $\signal{b}:=[b_1,\ldots,b_N]$. With this normalization, the solution $(\signal{p}^\star,c^\star)$ to Problem~\ref{problem.pro_fairness_canonical} has the following interpretation. The utility $c^\star$ is the highest rate observed in the network, and the weight $w_j\in~]0,1]$ is the fraction of this maximum rate that is assigned to user $j$. As described above, for $\bar{p}=1$, all users transmit with the same fraction of the rates they could achieve if alone in the system.

The above discussion shows that the sum throughput of the network is $\|\signal{w}\|_1 c^\star$. To evaluate the energy efficiency of the solution to Problem~\ref{problem.pro_fairness_canonical} with these settings, for given power budget $\bar{p}$, we use the $l_\infty$-energy function scaled by $\|\signal{w}\|_1$; i.e., we use $\|\signal{w}\|_1 E(\bar{p})$, where $E$ is the $\|\cdot\|_\infty$-energy efficiency function.  This scaled version of the $l_\infty$-energy efficiency, which has dimension bits/Joule, shows the ratio between the sum throughput in the network and the maximum observed transmit power $\bar{p}$ of a user. (We could use also the $l_1$-energy efficiency, in which case the energy efficiency is the rate achieved by a user for given total transmit power in the network.)

We approximate the solution to Problem~\ref{problem.pro_fairness_canonical} by applying 5,000 iterations of the fixed point algorithm in \refeq{eq.krause_iter} with the mapping $T$ in \refeq{eq.cont_ext}. The starting power allocation of the iterations is $\signal{p}_1=\signal{0}$. Fig.~\ref{fig:prop_div} and Fig.~\ref{fig:prop_div_c} show the energy efficiency and the utility (i.e., the rates) obtained with the simulation. For the bounds in Proposition~\ref{proposition.monotone_ee}, we use $\alpha_1=\alpha_2=\beta=1$. We can see that the results are consistent with the analysis in Sect.~\ref{sect.ee}. Even with a nonlinear mapping such as that in \refeq{eq.cont_ext}, we can clearly identify the two power regimes described in Sect.~\ref{sect.illustration}. In particular, with a power budget above $\|T(\signal{0})\|_\infty/\rho(\signal{M})$, which characterizes the high power regime, we note that increasing the power budget by orders of magnitude improves the achievable rates only marginally. We also observe a gap in the rate and energy efficiency bounds in the high SINR regime, which is expected because in this regime the bounds derived here are not necessarily asymptotically sharp for utility maximization problems with arbitrary interference mappings.

\begin{figure}
	\begin{center}
		\includegraphics[width=\figsize]{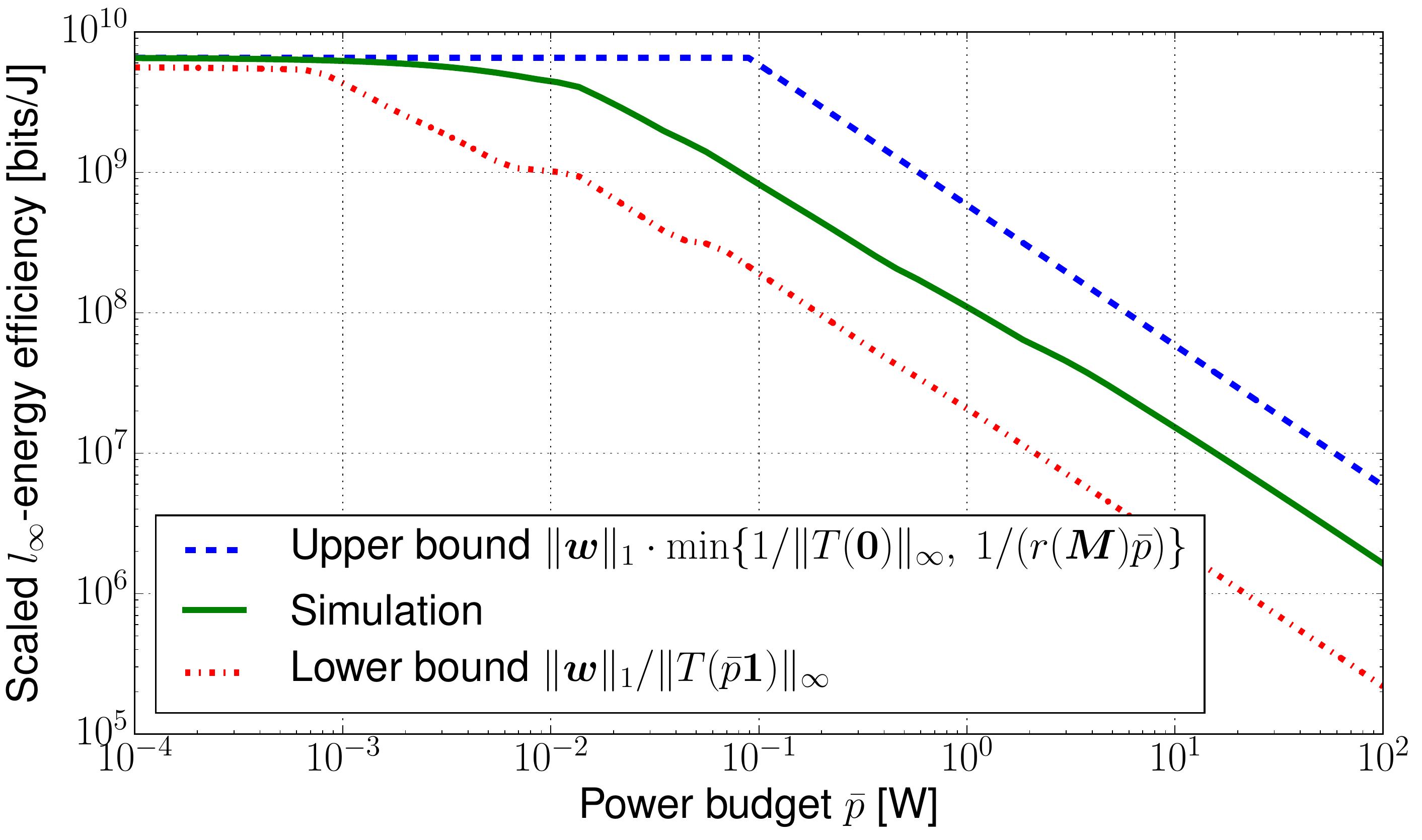}
		\caption{Energy efficiency as a function of the power budget $\bar{p}$.}
		\label{fig:prop_div}
	\end{center}
	\vspace{-0.5cm}
\end{figure}

\begin{figure}
	\begin{center}
		\includegraphics[width=\figsize]{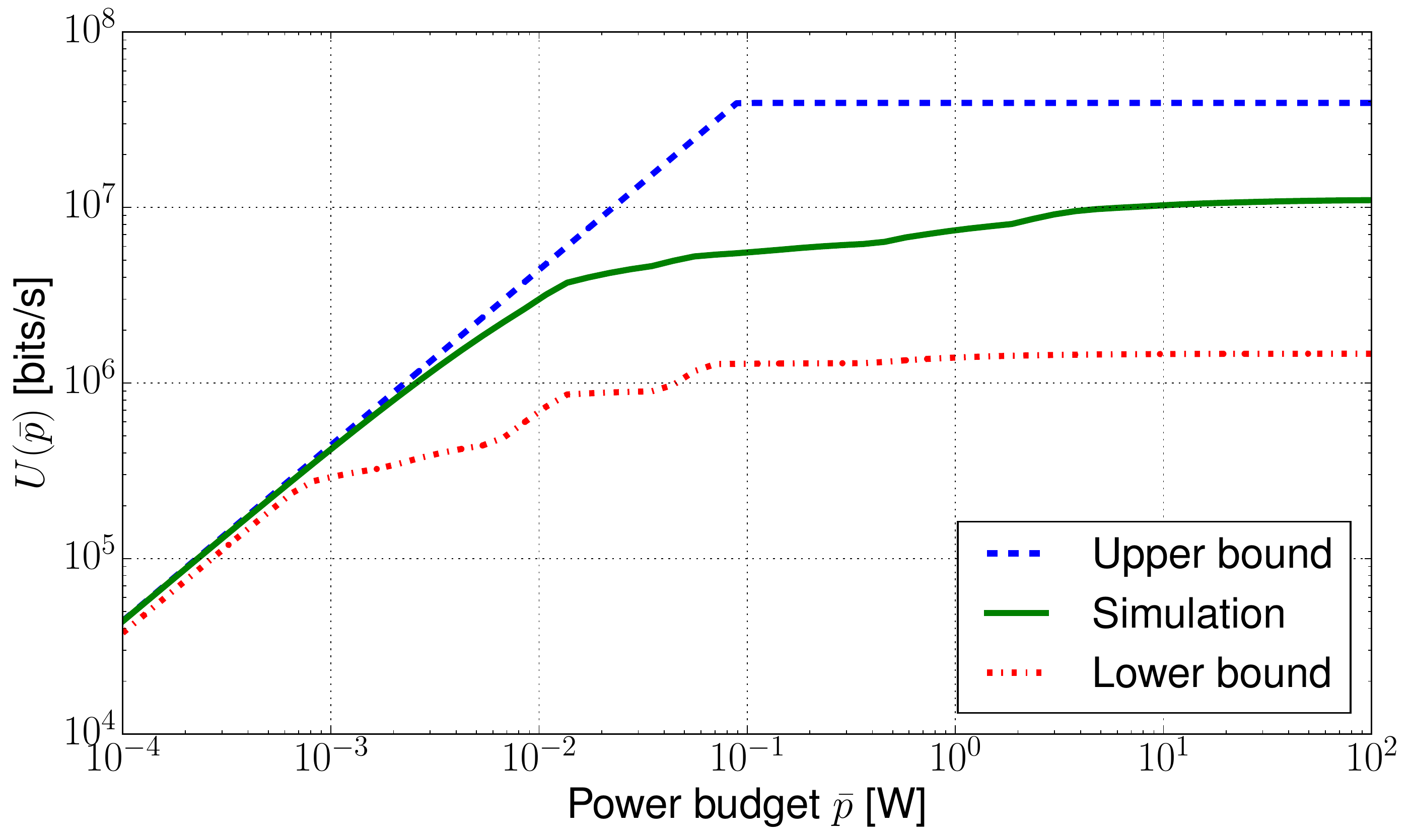}
		\caption{\small Highest rate as a function of the power budget .}
		\label{fig:prop_div_c}
	\end{center}
	\vspace{-0.5cm}
\end{figure}

\section{Summary and conclusions}
We have proved that the energy efficiency and the utility (e.g., rates or SINR) of solutions to a large class of network utility maximization problems are continuous and monotonic as a function of the power budget $\bar{p}$ available to the network. Furthermore, we used the concept of lower bounding matrices introduced in \cite{renato2016}, or, more generally, the concept of asymptotic or recession functions in convex analysis \cite{aus03}, to derive simple upper and lower bounds for the energy efficiency and for the network utility. In particular, the upper bounds reveal that the solutions are characterized by a low power regime and a high power regime, and the transition point is precisely known. In the low power regime, the upper bounds are asymptotically sharp (i.e., as the power budget tends to zero). The energy efficiency can decrease slowly as the power budget increases, whereas the utility grows linearly at best as a function of the power budget. In the high power regime, the typical behavior of interference-limited networks is that there are marginal gains in network utility and a fast decrease in energy efficiency as the power budget tends to infinity. In addition, the upper bounds we derived here are asymptotically sharp as the power budget diverges to infinity for the important family of network utility maximization problems constructed with affine interference mappings. 

The general theory developed here was illustrated with a novel joint uplink power control and base station assignment problem for weighted rate allocation. One of the main advantages of the  formulation is that it works directly with the rate of users, which is the parameter that network engineers are typically interested in maximizing. We showed that this problem can be solved optimally with an existing iterative method that, from a mathematical perspective, is nothing but a fixed point algorithm that solves a conditional eigenvalue problem. In this application, simulations show that the bounds are particularly good in the low power regime, and they are within the same order of magnitude of the optimal values in the high power regime. Therefore, the bounds derived here can serve as a simple estimate of the limits of a given network configuration. This fact can be especially useful in planning tasks that require the evaluation of multiple candidate configurations. With the bounds derived here, many inefficient configurations can be quickly ruled out without solving optimization problems.

\bibliographystyle{IEEEtran}
\bibliography{IEEEabrv,references}
\end{document}